\documentclass[onefignum,onetabnum]{siamonline171218}

\usepackage{amsfonts}
\usepackage{graphicx}
\usepackage{epstopdf}
\usepackage{algorithmic}
\ifpdf
  \DeclareGraphicsExtensions{.eps,.pdf,.png,.jpg}
\else
  \DeclareGraphicsExtensions{.eps}
\fi

\usepackage[]{amsmath}
\usepackage[]{amssymb} 
\usepackage{color} 
\usepackage{mathrsfs} 
\usepackage{subcaption}
\usepackage{cancel}
\usepackage{todonotes}

\usepackage{anysize}

\newsiamremark{remark}{Remark}
\newsiamremark{hypothesis}{Hypothesis}
\newsiamremark{assumptions}{Assumptions}
\crefname{hypothesis}{Hypothesis}{Hypotheses}
\newsiamthm{claim}{Claim}

\headers{Optimal Trading with Signals and Stochastic Price Impact}{J.-P. Fouque, S. Jaimungal and Y. F. Saporito}

\title{Optimal Trading with Signals and Stochastic Price Impact\thanks{{J-PF was supported by NSF grant DMS-1814091. SJ would like to acknowledge the support of the Natural Sciences and Engineering Research Council of Canada,
funding reference numbers RGPIN-2018-05705 and RGPAS-2018-522715.}}}

\author{Jean-Pierre Fouque\thanks{Department of Statistics and Applied Probability, University of California, Santa Barbara, United States 
 (\email{fouque@pstat.ucsb.edu}, \url{http://fouque.faculty.pstat.ucsb.edu/}).}
\and Sebastian Jaimungal\thanks{Department of Statistical Sciences, University of Toronto, Canada 
 (\email{sebastian.jaimungal@utoronto.ca}, \url{http://sebastian.statistics.utoronto.ca/}).}
\and Yuri F. Saporito\thanks{School of Applied Mathematics (EMAp), Getulio Vargas Foundation (FGV), Brazil
 (\email{yuri.saporito@fgv.br}, , \url{http://yurisaporito.com/}).}}

\usepackage{amsopn}

\newcommand{\ds}{\displaystyle} 
\newcommand{\eps}{\varepsilon} 

\newcommand{\cB}{\mathcal{B}}

\newcommand{\cA}{\mathcal{A}}

\newcommand{\cL}{\mathcal{L}}
\newcommand{\cH}{\mathcal{H}}
\newcommand{\cX}{\mathcal{X}}

\newcommand{\bR}{\mathbb{R}}

\newcommand{\bA}{\mathbb{A}}
\newcommand{\bE}{\mathbb{E}}
\newcommand{\bB}{\mathbb{B}}

\newcommand{\bfmu}{\boldsymbol{\mu}}
\newcommand{\bfrho}{\boldsymbol{\rho}}
\newcommand{\bfW}{\boldsymbol{W}}
\newcommand{\bfgamma}{\boldsymbol{\gamma}}
\newcommand{\bfb}{\mathbf{b}}
\newcommand{\bfV}{\mathbf{V}}

\allowdisplaybreaks

\usepackage{wrapfig}
\usepackage{dsfont}
\usepackage{enumitem}

\newcommand{\seb}[1]{{\color{black} #1}}

\newcommand{\yuri}[1]{{\color{black} #1}}

\ifpdf
\hypersetup{
  pdftitle={Optimal Trading with Signals and Stochastic Price Impact},
  pdfauthor={J.-P. Fouque, S. Jaimungal and Y. F. Saporito}
}
\fi

\renewcommand{\P}{\mathbb{P}}
\newcommand{\F}{\mathcal{F}}

\begin{document}

\maketitle

\begin{abstract}
Trading frictions are stochastic. They are, moreover, in many instances fast-mean reverting. Here, we study how to optimally trade in a market with stochastic price impact and study approximations to the resulting optimal control problem using singular perturbation methods. We prove, by constructing sub- and super-solutions, that the approximations are accurate to the specified order. Finally, we perform some numerical experiments to illustrate the effect that stochastic trading frictions have on optimal trading.
\end{abstract}



\section{Introduction}

Trading frictions exist across all electronic markets and stem from a multitude of factors. In quote driven markets, liquidity providers place  limit orders to buy and sell assets, while liquidity takers use market orders to extract this liquidity, sometimes by walking the limit order book (LOB), but often by picking off available volume at the touch (the volume at the best bid/ask prices). In either case, trading quickly (or with large volume) induces a trading cost while trading slowly introduces price risk. The seminal work of \cite{almgren2001optimal} study how to optimally trade in such environments. \cite{almgren2012optimal} extends the analysis to the case of stochastic liquidity and volatility and, from numerical methods applied to the resulting Hamilton-Jacobi-Bellman (HJB) equation, shows that stochastic liquidity is indeed an important component in optimal execution.

\begin{wrapfigure}{r}{0.4\textwidth}
    \centering
    \includegraphics[width=0.4\textwidth]{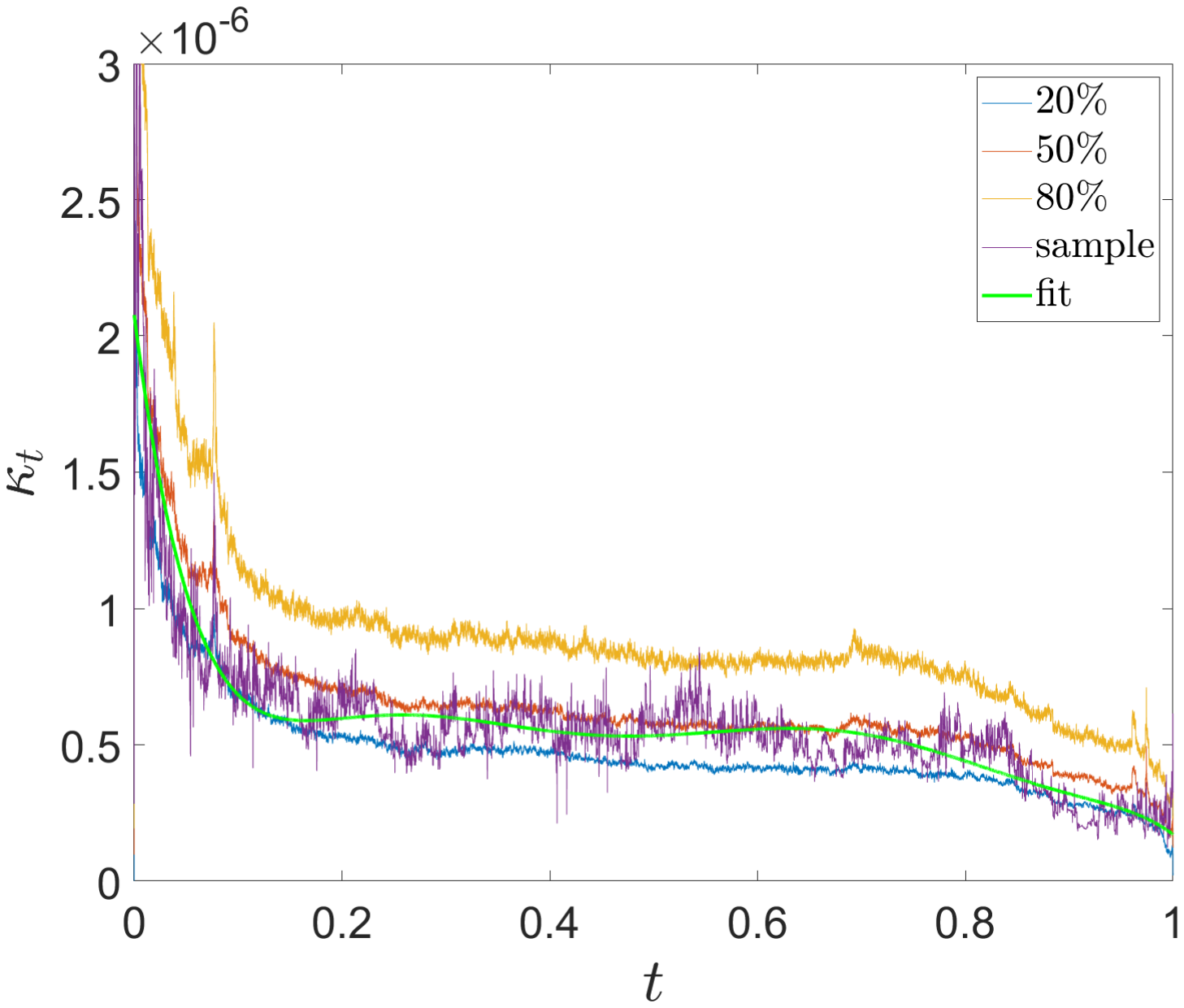}
    \caption{Sample path (May 28, 2014), two fits for $\kappa(t)$, and quantiles for $\kappa_t$ for MSFT using 2014 data. \seb{$t=0,1$ is the start,end of the trading day.} }
    \label{fig:kappa_fit}
\end{wrapfigure}
While the cost of trading fluctuates throughout the day, markets react quickly -- sometimes on the timescale of milliseconds or even microseconds.
The fast-mean-reverting aspect of the temporary price impact is documented in \cite{CarteaJaimungal2015}. For example, Figure \ref{fig:kappa_fit} shows (i) a sample path of trading impact (measured as the effective cost of walking the limit order book per unit of volume) denoted $\kappa_t$, (ii) a deterministic fit (denoted by $\kappa(t)$, see Section \ref{sec:real-data}) to the sample path, and (iii) the $20\%$, $50\%$, and $80\%$ quantiles through the day (across all of 2014 data) for Microsoft (MSFT).

Based on these observations, here we study how stochastic price impact driven by a fast mean-reverting process affects the optimal execution of trades. Fast mean-reverting (FMR) models have appeared widely in stochastic volatility models, as first proposed in \cite{fouque2000mean} and \cite{fouque2000derivatives}, \cite{fouque2003multiscale}, and extended and applied in numerous areas including stochastic interest models \cite{cotton2004stochastic}, commodity models \cite{hikspoors2008asymptotic,fouque2014multiscale}, credit risk \cite{papageorgiou2008multiscale}, and portfolio optimisation problems \cite{fouque2017portfolio}. \yuri{As far as the authors are aware, the only other reference that considers FMR models in the context of algorithmic trading is \cite{chan2015optimal}. However, they model volatility as FMR whereas we consider stochastic impact. Additionally, differently from our paper, they consider infinity time horizon setting and do not provide a rigorous proof of their approximation.}

\cite{graewe2018smooth} studies stochastic impact models that generalise those in \cite{almgren2012optimal} and develop a novel approach for looking at a general class of optimal control problems with singular terminal state constraints.  \cite{moazeni2013optimal} analyzes the effect of jumps in prices to model the uncertain price impact from other traders.
\cite{cheridito2014optimal} examines discrete-time optimal execution problem with stochastic volatility and liquidity and (for linear criterion) characterise the optimal strategy in terms of a certain forward-backward-stochastic-differential equation (FBSDE). \cite{barger2019optimal}, which is perhaps the closest work to ours, analyses stochastic impact by making use of a coefficient expansion of the resulting HJB equation and provide a formal expansion of the approximate strategy (but provides no proof of accuracy and does not utilise the FMR property).  A distinct line of work looks at modeling impact through the dynamics of the limit order book (LOB), where orders walk the LOB and after some time it recovers -- this is called the resilience of the LOB. \cite{obizhaeva2013optimal} is the earliest work along these lines, and many others generalized the approach in various directions. For example, \cite{fruth2019optimal} studies a stochastic ``illiquidity'' process that induces costs when the trader's position jumps, \cite{siu2019optimal} look at regime switching driven market resilience. For an overview of optimal execution, and more generally algorithmic trading and market microstructure, see, e.g., \cite{cartea2015algorithmic}, \cite{gueant2016financial}, and
\cite{laruelle2018market}.

We distinguish our work from the extant literature in several ways. First, to the authors knowledge, this is the first example of an approximate closed-from solution for stochastic impact that reflects the FMR behaviour observed in the market. Second, we include not only stochastic impact, but also stochastic trading signals (which may be seen as stemming from order-flow or other factors). Thirdly, we provide proof that the approximate optimal strategy are indeed approximately optimal to the order specified.

The remainder of this paper is organized as follows: Section \ref{sec:main} introduces our framework, performs a dimensional reduction, and introduces the approximate optimal strategy whose accuracy is addressed later. Section \ref{sec:MR-trading-signal} provides an explicit trading signal model where the approximate optimal strategy may be computed exactly. Section \ref{sec:real-data} shows how to estimate parameters from data and applies the approximate optimal strategy in a simulation environment. Section \ref{sec:approx_proof} provides a proof accuracy of the approximation. 

\section{Approximate Optimal Strategy}
\label{sec:main}

In this paper, we investigate optimal trading in environments where price impact is stochastic and driven by a fast mean-reverting factor and the trader incorporates price signals. This generalizes the results in  \cite{CarteaJaimungal2015} for optimal execution with linear temporary and permanent price impacts as well as those studied in \cite{lehalle2019incorporating} -- who looks at particular cases of \cite{CarteaJaimungal2015} \seb{and reinterpret the stochastic order-flow model developed in \cite{CarteaJaimungal2015}  as a trading signal. Indeed, several others including \cite{donnelly2020optimal,belak2018optimal,forde2021optimal} investigate how trading signals modify the optimal strategy, but none, as far as we are aware, also include stochastic impact.}.

Specifically, let $(\Omega,\P,\F=(\F_t)_{t\ge0})$ denote a complete filtered probability space. What generates the filtration is  specified shortly. On this probability space, we also introduce a number of stochastic processes: the asset price process $S^\nu=(S_t^\nu)_{t\ge0}$, the trader's  cumulative cash process $X^\nu=(X_t^\nu)_{t\ge0}$, the trader's cumulative inventory $Q^\nu=(Q_t^\nu)_{t\ge0}$, the trading signals $\bfmu=(\mu_t^1,\dots,\mu_t^d)_{t\ge0}$, and the driver of stochastic impact $Y^\eps=(Y^\eps_t)_{t\ge0}$. The process $\nu=(\nu_t)_{t\ge0}$ denotes the trader's rate of trading, and this is what the trader controls -- hence the superscript on the various processes that are directly affected by the trader's actions. The various stochastic processes are assumed to satisfy the system of SDEs:
\begin{align}\label{eq:model_fast}
\begin{cases}
dS_t^\nu = (\bfgamma \cdot \bfmu_t + b\,\nu_t)\,dt + \sigma \,dW_t,
\\[5pt]
\ds dX_t^\nu = -\left(S_t^\nu + k(t,Y_t^\eps)\nu_t \right)\nu_t\, dt,
\\[5pt]
dQ_t^\nu = \nu_t\, dt,\\[5pt]
d\bfmu_t =c(\bfmu_t)\,dt + g(\bfmu_t)\,d\bfW'_t,\quad \text{and}
\\[5pt]
\ds dY_t^\eps = \frac{1}{\eps}\, \alpha(Y_t^\eps)\,dt + \frac{1}{\sqrt{\eps}} \beta(Y_t^\eps)\, dW_t^*\;.
\end{cases}
\end{align}
Here, $b$ is the permanent price impact and $k(t,y)$ is the temporary price impact which is driven by the fast mean-reverting process $Y^\eps$. The processes $W$, $\bfW'$ and $W^*$ are correlated Brownian motions, where both $W$ and $W^*$ are one-dimensional and $\bfW'$ is $d$-dimensional. Moreover, we assume $d[W'_i,W^*]_t = \rho_i\, dt$, $i\in\{1,\dots,d\}$, where $\bfrho:=(\rho_1,\dots,\rho_d) \in (-1,1)^d$ and  the joint correlation structure of the three Brownian motions is positive definite.
The filtration $\F=(\F_t)_{t\ge0}$ is taken to be the natural filtration induced by these Brownian motions. \yuri{The functions $k$, $g$ and $\beta$ are assumed to be non-negative. Additionally, the functions $\alpha$ and $\beta$ that define the dynamics of $Y^\eps$ are such that it has a unique stationary distribution independent of $\eps$.}

The dynamics in \eqref{eq:model_fast} may be interpreted as follows. In the absence of the trader's actions, the asset price process has a predictable drift $\bfgamma \cdot \bfmu$ that stems for the (vector valued)  trading signal $\bfmu$. The trader's action, however, pushes prices in the direction of their trades. In the current setup, we view $\nu_t>0$ as buying, therefore while the trader sends buy/sell orders they push the price upwards/downwards. Noise traders cause prices to fluctuate, and those fluctuations are captured by the Brownian motion $W$. The trader's cumulative cash process increases when the sell, but they do not receive the midprice $S_t^\nu$, rather they incur additional liquidity costs. Those costs are approximated as being linear in the trading speed $\nu$, and the coefficient of the impact is given by $k(t,Y_t^\eps)$ -- which incorporates the fast mean-reverting nature of the trading frictions, as well as its diurnal pattern as observed in, e.g. Figure \ref{fig:kappa_fit}. To model this pattern, we  assume that
\begin{align}
k(t,y) = \kappa(t) \left(1 + \eta(y)\right)^{-1},
\end{align}
$\kappa(t) > 0$, $t\in[0,T]$, and that $\eta > -1$.  \yuri{We chose to model the stochastic behavior of the temporary price impact as above in order to incorporate that the fact the variation is around a given mean function, denoted here by $\kappa$, and that $k(t,y)$ appears in the PDE of interest here as its inverse, cf. \eqref{eqn:PDE-after-feedback}}.
We denote the average with respect to the stationary distribution of $Y^\eps$ by $\langle \cdot \rangle$ and, without loss of generality, we impose the constraint that $\left\langle \eta \right\rangle = 0$.

\seb{To gain some intuition on the role that the correlation $\bfrho$ of $\bfW'$ and $W^*$ plays, consider the covariation of $k=(k(t,Y_t^\eps))_{t\ge0}$ and $\bfmu$. To this end, from Ito's lemma we have
\begin{align*}
dk_t =&\; -\frac{\kappa(t)}{(1+\eta(Y_t^\eps))^2}\,\eta'(Y_t^\eps)\,dY_t^\eps \\
&\;\;+ \tfrac12 \kappa(t)\Big((1+\eta(Y_t^\eps))^{-3}\,(\eta'(Y_t^\eps))^2-(1+\eta(Y_t^\eps))^{-2}\,\eta''(Y_t^\eps)\Big) d[Y^\eps,Y^\eps]_t
\end{align*}}
\seb{so that}
$$\color{black}
d[k,\bfmu]_t=-\frac{\eta'(Y_t^\eps)}{1+\eta(Y_t^\eps)}\,k_t\,g(\bfmu_t)\,\frac{1}{\sqrt{\eps}}\,\beta(Y_t^\eps)\,\bfrho\,dt.
$$
\seb{Hence, the instantaneous correlation of $k$ and $\bfmu$ is $\bfrho_{k,\bfmu}=-\frac{\eta'(Y_t^\eps)}{1+\eta(Y_t^\eps)}\,k_t\,\bfrho$, and its sign is \allowbreak $-\mathrm{sgn}(\eta'(Y_t^\eps)\,\bfrho)$. In the numerical investigations in \Cref{sec:real-data}, we set $\eta(y)=y$, so that the sign of the instantaneous correlation is determined by $-\bfrho$.}

The investor faces the following optimal control problem:
\begin{align}
H^\eps(t,x,S,\bfmu,q,y) &= \sup_{\nu \in \bA} \cH^\nu(t,x,S,\bfmu,q,y),
\quad \text{where}
\\
\cH^\nu(t,x,S,\bfmu,q,y) &= \bE_{t,S,q,x,\bfmu,y}\left[X_T^\nu + Q_T^\nu\left(S_T^\nu - \varphi \, Q_T^\nu \right) - \phi \int_t^T (Q_u^\nu)^2 du \right].\label{eqn:performance-criterion}
\end{align}
In \eqref{eqn:performance-criterion}, $\phi$ is an inventory penalty parameter (and may be interpreted as stemming from a penalty on quadratic variation of the book-value of the traders position, or, more interestingly, from model uncertainty \cite{cartea2017algorithmic}); $\varphi$ is a liquidation penalty parameter which, if large, forces the trader to end with little or no inventory (see \cite{graewe2018smooth} for how to enforce exact liquidation with a singular terminal condition); $\bE_{t,S,q,x,\bfmu,y}[\cdot]$ denotes  conditional expectation given $S_t^\nu = S$, $Q_t^\nu = q$, $X_t^\nu = x$, $\bfmu_t = \bfmu$ and $Y_t^\eps = y$
; and $\bA$ is the set of admissible controls that consist of $\mathcal{F}$-predictable processes  in $\mathbb{L}^2([0,T],\Omega)$, i.e. $\mathbb E[\int_0^T |\nu_t|^2\,dt]<+\infty$, taking values in $\cA:=\mathds{R}$.

Note, the investor's wealth process and the asset price can be removed from the optimization problem.  Firstly, notice that, by the product rule (and that $Q^\nu$ is absolutely continuous), we have
\begin{align}
\begin{split}
X_T^\nu + Q_T^\nu S_T^\nu &= x -\int_t^T \left(S_u^\nu + k(u,Y_u^\eps)\nu_u \right)\nu_u\, du + qS + \int_t^T S_u^\nu dQ_u^\nu + \int_t^T Q_u^\nu dS_u^\nu \\
&= x + qS - \int_t^T S_u^\nu \nu_u du - \int_t^T k(u,Y_u^\eps) \,\nu_u^2 du
+ \int_t^T S_u^\nu\nu_u du  + \int_t^T Q_u^\nu dS_u^\nu
\end{split}\\
&= x + qS - \int_t^T k(u,Y_u^\eps) \nu_u^2 du + \int_t^T Q_u^\nu dS_u^\nu.
\end{align}
Hence, we may write
\begin{equation}
\begin{split}
    \cH^\nu(t,x,S,\bfmu,q,y) = x &+ qS\,
+  \,\bE_{t,S,q,x,\bfmu,y}\left[- \varphi (Q_T^\nu)^2 - \int_t^T k(u,Y_u^\eps) \,\nu_u^2\, du
\right.
\\
&\left.\hspace*{9em}  + \int_t^T Q_u^\nu\, (\bfgamma \cdot \bfmu_u + b\nu_u)\,du - \phi \int_t^T (Q_u^\nu)^2 \,du \right].
\end{split}
\label{eqn:H-with-x-and-qS}
\end{equation}
Next, as  $(\bfmu_t, Q_t^\nu, Y_u^\eps)$ is Markovian, we may define
\begin{align}
Z(t,\bfmu,q,y,v) &= -k(t,y) v^2 + q (\bfgamma \cdot \bfmu + bv) - \phi q^2, \\
\mathcal{Z}^{\nu}(t,\bfmu,q,y) &= \bE_{t,q,\bfmu,y}\left[-\varphi(Q_T^\nu)^2 + \int_t^T Z(u,\bfmu_u,Q_u^\nu,Y_u^\eps,\nu_u)du \right], \label{eq:mathcal_Z}
\end{align}
and
\begin{align}\label{eq:h_eps}
h^\eps(t,\bfmu,q,y) = \sup_{\nu \in \bA} \mathcal{Z}^{\nu}(t,\mu,q,y),
\end{align}
where $\bE_{t,q,\bfmu,y}[\cdot]$ here means the conditional expectation given $Q_t^\nu = q$, $\bfmu_t = \bfmu$, and $Y_t^\eps = y$. Hence, comparing \eqref{eqn:H-with-x-and-qS} with \eqref{eq:h_eps}, we may write\footnote{\yuri{One should notice that we are not following the usual argument of proposing an ansatz for $H^\eps$ of the form \eqref{eqn:first-ansatz-for-H} but actually showing that $H^\eps$ is indeed of that form.}}
\begin{align}
H^\eps(t,x,S,\bfmu,q,y) = x + qS + h^\eps(t,\bfmu,q,y).
\label{eqn:first-ansatz-for-H}
\end{align}
From the dynamic programming principle, we expect that $h^\eps$ satisfies the associated Hamilton-Jacobi-Bellman (HJB) equation
\begin{multline}\label{eq:pdf_h_eps}
\partial_t h^\eps
+ \frac{1}{\eps} \cL_0 h^\eps + \cL_{\bfmu} h^\eps - \phi q^2  + \frac{1}{\sqrt{\eps}} \beta(y) \sum_{i=1}^k (g_i(\bfmu) \cdot \bfrho )\partial_{\mu_i y} h^\eps
\\
+ (\bfgamma \cdot \bfmu)q + \sup_{v \in \cA} \left\{  (bq + \partial_q h^\eps)v   - k(t,y)v^2 \right\} = 0,
\qquad\qquad
\end{multline}
with terminal condition $h^\eps(T,\bfmu,q,y) = -\varphi \,q^2$, where $g_i(\bfmu)$ is the $i$-th row of $g(\bfmu)$. Here, $\cL_0$ and $\cL_{\bfmu}$ are the infinitesimal generators
\begin{align}
\cL_0 &= \alpha(y) \partial_y + \tfrac{1}{2}\beta^2(y) \partial_{yy},
\qquad \text{and} \\
\cL_{\bfmu}&= c(\bfmu) \cdot \partial_{\bfmu} + \tfrac{1}{2} \mbox{Tr}((g(\bfmu) g(\bfmu)^T)\partial_{\bfmu \bfmu}),
\end{align}
respectively.

This HJB equation is not solvable in closed-form for a multitude of reasons. Rather, in the spirit of \cite{fouque_portfolio_selection} for portfolio optimization problems and elaborated in \cite{multiscale_fouque_new_book},  we aim to obtain an approximate solution that holds when the stochastic  price impact driving factor $Y^\eps$ is fast mean-reverting. We also prove that the approximation we propose is correct to the appropriate order, and we do so using the following broad steps:
\begin{enumerate}[label=({\roman*})]
    \item determine the feedback form of the optimal control and derive the non-linear PDE that $h^\eps$ satisfies,

    \item notice that $h^\eps$ may be represented as a quadratic form in $q$, but \yuri{with possibly more complex non-linearity} in all other state variables,

    \item perform a formal expansion of each coefficient of $h^\eps$ in powers of $\eps$,

    \item notice that constant (in $q$) and linear  (in $q$) functions satisfy linear PDEs and apply, now classical, methods from \cite{multiscale_fouque_new_book} to prove accuracy, and finally

    \item the quadratic (in $q$) function satisfies a non-linear PDE, and for this we develop a super and sub-solution approach to
prove that the approximation error is controlled.
\end{enumerate}
\vspace*{1em}

Proceeding along the above lines, we first apply the first-order condition (FOC) to obtain the optimal feedback control as
\begin{align}
\nu_*^\eps(t,\bfmu,q,y) = \frac{b\,q + \partial_q h^\eps(t,\bfmu,q,y) }{2\,k(t,y)},
\end{align}
and, upon substitution of the FOC, the HJB equation \eqref{eq:pdf_h_eps} reduces to
\begin{multline}
\left(\partial_t + \cL_{\bfmu} \right) h^\eps - \phi\, q^2  + (\bfgamma \cdot \bfmu)\, q  + \frac{\left(b\,q + \partial_q h^\eps \right)^2}{4\,k(t,y)}
\\
+ \frac{1}{\sqrt{\eps}} \beta(y) \sum_{i=1}^k (g_i(\bfmu) \cdot \bfrho )\partial_{\mu_i y} h^\eps + \frac{1}{\eps} \,\cL_0 h^\eps  =0\,,
\hspace*{5em}
\label{eqn:PDE-after-feedback}
\end{multline}
where, for readability, we suppress the dependence of $h^\eps$ on its arguments.

We define a non-linear differential operator $\cL(k)$ that corresponds to the temporary price impact $k$, and it acts on functions $f$ as follows
\begin{align}
\cL(k)[f] := \left(\partial_t + \cL_{\bfmu} \right)f - \phi \,q^2 + (\bfgamma \cdot \bfmu) \,q   + \frac{\left(b\,q + \partial_qf\right)^2}{4\,k}\,.
\end{align}
This operator is, in general, non-linear due to the last term.
With this notation, we may write the PDE \eqref{eqn:PDE-after-feedback} as
\begin{align}\label{eq:final_pde}
&\cL(k(t,y))[h^\eps] + \left(\frac{1}{\sqrt{\eps}} \cL_1 \ + \frac{1}{\eps} \cL_0\right) h^\eps  = 0\;,
\end{align}
where
\begin{align}
\cL_1 &= \beta(y) \sum_{i=1}^k (g_i(\bfmu) \cdot \bfrho )\partial_{\mu_i y}\;.
\end{align}

\subsection{Expansion in powers of $q$}

As in the case with constant price impact studied in \cite{CarteaJaimungal2015}, but including trading signals, we may write $h^\eps$ as a quadratic polynomial in $q$ without any loss of generality. This is due to the form of the terminal conditions and the coefficients of the PDE \eqref{eq:final_pde}. To this end, we write
\begin{align}
h^\eps(t,\bfmu,q,y) = h^{(0),\eps}(t,\bfmu,y) + h^{(1),\eps}(t,\bfmu,y)\,q + h^{(2),\eps}(t,\bfmu,y)\, q^2
\label{eqn:h-expansion-in-q}
\end{align}
and aim to determine approximations for  $(h^{(i),\eps})_{i=1,2,3}$.
Inserting \eqref{eqn:h-expansion-in-q}  into  \eqref{eq:final_pde}, collecting powers of $q$, and setting the coefficient of each power  to zero, we find the following coupled system of PDEs for $(h^{(i),\eps})_{i=1,2,3}$
\begin{align}
&\begin{cases}
\displaystyle
(\partial_t + \cL_{\bfmu}) h^{(2),\eps} - \phi + \frac{(b + 2h^{(2),\eps})^2}{4k(t,y)} + \frac{1}{\sqrt{\eps}} \cL_1 h^{(2),\eps} + \frac{1}{\eps} \cL_0 h^{(2),\eps} = 0\,,
\\[0.75em]
h^{(2),\eps}(T,\bfmu,y) = - \varphi\,,
\end{cases} \label{eq:pde_h2eps}
\\[1em]
&\begin{cases}
\displaystyle
(\partial_t + \cL_{\bfmu}) h^{(1),\eps} + \bfgamma \cdot \bfmu + \frac{(b + 2h^{(2),\eps})h^{(1),\eps}}{2\,k(t,y)} + \frac{1}{\sqrt{\eps}} \cL_1 h^{(1),\eps} + \frac{1}{\eps} \cL_0 h^{(1),\eps} = 0\,,
\\[0.75em]
h^{(1),\eps}(T,\bfmu,y) = 0\,,
\end{cases} \label{eq:pde_h1eps}
\\[1em]
&\begin{cases}
\displaystyle
(\partial_t + \cL_{\bfmu}) h^{(0),\eps} + \frac{(h^{(1),\eps})^2}{4k(t,y)} + \frac{1}{\sqrt{\eps}} \cL_1 h^{(0),\eps} + \frac{1}{\eps} \cL_0 h^{(0),\eps} = 0\,,
\\[0.75em]
h^{(0),\eps}(T,\bfmu,y) = 0\,.
\end{cases} \label{eq:pde_h0eps}%
\end{align}%
The PDEs \eqref{eq:pde_h0eps} and \eqref{eq:pde_h1eps} for $h^{(0),\eps}$ and $h^{(1),\eps}$ are linear. They do, however,  have non-linear potential terms. These potentials are ``known'' when solving the problem sequentially in the order $h^{(2),\eps}$ then $h^{(1),\eps}$ and finally $h^{(0),\eps}$.  Contrasting, the PDE \eqref{eq:pde_h2eps} for $h^{(2),\eps}$ is non-linear. Similarly to the constant price impact case studied in \cite{CarteaJaimungal2015}, \eqref{eq:pde_h2eps} contains no coefficients or terminal conditions that depend on $\bfmu$, hence the solution $h^{(2),\eps}$ is not a function of $\bfmu$, and we write $h^{(2),\eps}(t,y)$. Performing a constant shift
\begin{align}
\cX^\eps(t,y) := h^{(2),\eps}(t,y) + \frac{b}{2}
\end{align}
we may write \eqref{eq:pde_h2eps} and \eqref{eq:pde_h1eps}  as (the PDE  \eqref{eq:pde_h0eps} is unmodified)
\begin{align}
&\begin{cases}
\displaystyle \partial_t \cX^\eps - \phi + \frac{(\cX^\eps)^2}{k(t,y)} + \frac{1}{\eps} \cL_0 \cX^\eps = 0\,,
\\[0.75em]
\cX^\eps(T,y) = - \varphi + b/2\,,
\end{cases} \label{eq:pde_chieps}
\\[1em]
&\begin{cases}
\displaystyle (\partial_t + \cL_{\bfmu}) h^{(1),\eps} + \bfgamma \cdot \bfmu + \frac{\cX^\eps h^{(1),\eps}}{k(t,y)} + \frac{1}{\sqrt{\eps}} \cL_1 h^{(1),\eps} + \frac{1}{\eps} \cL_0 h^{(1),\eps} = 0\,,
\\[0.75em]
h^{(1),\eps}(T,\bfmu,y) = 0\,.
\end{cases}\label{eq:pde_h1_eps_new}
\end{align}

\subsection{The case of deterministic temporary price impact}\label{sec:deterministic_k}

The case of constant $k$ is fully characterized in \cite{CarteaJaimungal2015} and it is straightforward to generalize that result when $k$ is a deterministic function of time: $k(t,\cdot) \equiv \mathsf{k}(t)$. Indeed, consider the PDE
\begin{align}
\begin{cases}
\left(\partial_t + \cL_{\bfmu} \right)h  - \phi \,q^2 + (\bfgamma \cdot \bfmu) \,q   + \frac{\left(b\,q + \partial_qh\right)^2}{4\,\mathsf{k}(t)} = 0,\\
h(T,\bfmu,q) =-\varphi q^2.
\end{cases}
\end{align}
Next, $h$ may be written as
\begin{align}\label{eq:h_deterministic}
h(t,\bfmu,q) &= h^{(0)}(t,\bfmu) +  h^{(1)}(t,\bfmu) q +  h^{(2)}(t)q^2,
\end{align}
where $\cX(t) = h^{(2)}(t) + b/2$ satisfies the Ricatti equation
\begin{align}
\label{eqn:Ricatti-deterministic}
\left\{
\begin{array}{rl}
\displaystyle \cX'(t) - \phi + \frac{1}{\mathsf{k}(t)}\cX^2(t) &= 0,\\
\cX(T) &= -\varphi + b/2,
\end{array}
\right.
\end{align}
and $h^{(1)}, h^{(0)}$ are given by
\begin{align}
h^{(1)}(t,\bfmu) &=  \int_t^T e^{\int_t^s \frac{\cX(u)}{\mathsf{k}(u)} du}
\;\bE\left[\, \bfgamma \cdot \bfmu_s \  | \ \bfmu_t = \bfmu \ \right] ds,
\quad \text{and}
\label{eq:h1cj}
\\
h^{(0)}(t,\bfmu) &= \int_t^T \frac{1}{4\,\mathsf{k}(s)} \;\bE\left[ \ (h^{(1)})^2(s,\bfmu_s) \ | \ \bfmu_t = \bfmu \ \right]ds. \label{eq:h0cj}
\end{align}
See Appendix \ref{app:proofs_time_dep} for the derivation of these formulas. However, no closed-form solution for $\cX$ is available in this case.


\subsection{Perturbation Framework}\label{sec:perturbation}

Motivated by  the perturbation framework of \cite{multiscale_fouque_new_book}, we first carry out formal expansions and provide the rigorous proof of accuracy of the approximation  in Section \ref{sec:approx_proof}.

To proceed with the approximations,  consider the formal expansion of $\cX^\eps$ in powers\footnote{\yuri{It is not necessary to consider the formal expansion in powers of $\sqrt{\eps}$ because the PDE \eqref{eq:pde_chieps} is affected only by $\tfrac{1}{\eps}\cL_0$ whilst the PDEs \eqref{eq:pde_h0eps} and \eqref{eq:pde_h1_eps_new} are affected by both $\tfrac{1}{\eps}\cL_0$ and $\tfrac{1}{\sqrt{\eps}}\cL_1$.}} of $\eps$
\begin{align}
\cX^\eps = \cX_0 + \eps \cX_1 + \eps^2 \cX_2 + \cdots,
\end{align}
and of $h^{(i),\eps}$, $i=0,1$, in powers of $\sqrt{\eps}$:
\begin{align}
h^{(i),\eps} = h^{(i)}_0+ \sqrt{\eps} \ h^{(i)}_1 + \eps \ h^{(i)}_2 + \cdots.
\end{align}
Inserting these formal expansions into the PDEs (\ref{eq:pde_chieps}), (\ref{eq:pde_h1_eps_new}) and (\ref{eq:pde_h0eps}), collecting terms of like powers of $\eps$, and setting each to zero separately, we find, up to order $\sqrt{\eps}$, the following system of PDEs:
\begin{subequations}
\begin{align}
\cL_0 \cX_0 &= 0, \label{eq:chi-1order}
\\[0.25em]
\cL_0 \cX_1 + \partial_t \cX_0 - \phi + \frac{\cX_0^2}{k(t,y)} &= 0, \label{eq:chi0order}
\\[0.25em]
\cL_0 h^{(i)}_0 &= 0, \label{eq:hi-1order}
\\[0.25em]
\cL_1 h^{(i)}_0 + \cL_0 h^{(i)}_1 &= 0, \label{eq:hi-12order}
\\[0.25em]
(\partial_t + \cL_{\bfmu}) h^{(1)}_0 + \bfgamma \cdot \bfmu + \frac{\cX_0 h^{(1)}_0}{k(t,y)} + \cL_1 h^{(1)}_1 + \cL_0 h^{(1)}_2 &= 0, \label{eq:h10order}
\\[0.25em]
(\partial_t + \cL_{\bfmu}) h^{(0)}_0 + \frac{(h^{(1)}_0)^2}{4k(t,y)} + \cL_1 h^{(0)}_1 + \cL_0 h^{(0)}_2 &= 0, \label{eq:h00order}
\\[0.25em]
(\partial_t + \cL_{\bfmu}) h^{(1)}_1 + \frac{\cX_0 h^{(1)}_1}{k(t,y)} + \cL_1 h^{(1)}_2 + \cL_0 h^{(1)}_3 &= 0, \label{eq:h112order}
\\[0.25em]
(\partial_t + \cL_{\bfmu}) h^{(0)}_1 + \frac{h^{(1)}_0 h^{(1)}_1}{2\,k(t,y)} + \cL_1 h^{(0)}_2 + \cL_0 h^{(0)}_3 &= 0, \label{eq:h012order}
\end{align}%
\end{subequations}%
with terminal conditions $\cX_0(T,y) = - \varphi + b/2$, $h^{(i)}_0(T,\bfmu,y) = 0$ and $h^{(i)}_1(T,\bfmu,y)= 0$, $i=0,1$.

\begin{proposition}[Zero-order Terms]
A solution for the zero-order approximation terms $\cX_0$, $h^{(0)}_0$ and $h^{(1)}_0$ that solves PDEs (\ref{eq:chi-1order}), (\ref{eq:hi-1order}) and (\ref{eq:hi-12order}) and centers the Poisson equations (\ref{eq:chi0order}), (\ref{eq:h10order}) and (\ref{eq:h00order}) are given by Equations (\ref{eqn:Ricatti-deterministic}), (\ref{eq:h0cj}) and (\ref{eq:h1cj}), respectively, with deterministic temporary price impact $\mathsf{k}(t) = \kappa(t)$.
\end{proposition}

\begin{proof}

From the first equation (\ref{eq:chi-1order}), we choose $\cX_0 = \cX_0(t)$ independent of $y$. The second equation (\ref{eq:chi0order}) is a Poisson equation for $\cX_1$ and its centering condition is
\begin{equation*}
\left\langle\partial_t \cX_0 - \phi + \frac{\cX_0^2}{k(t,\cdot)} \right\rangle = 0 \quad \Leftrightarrow\quad
\partial_t \cX_0 - \phi + \frac{\cX_0^2}{\kappa(t)} = 0\ ,
\end{equation*}
because
\begin{align}
\left\langle \frac{1}{k(t,\cdot)}\right \rangle = \left\langle \frac{1}{\kappa(t)}(1 + \eta(\cdot))\right \rangle = \frac{1}{\kappa(t)}(1 + \langle \eta \rangle) =   \frac{1}{\kappa(t)},
\end{align}
where the last equality follows from the standing assumption $\langle \eta \rangle=0$. This is the same Ricatti ODE as \eqref{eqn:Ricatti-deterministic} with $\mathsf{k} = \kappa$ and therefore we may write $\cX_0(t) =  \cX(t) = h^{(2)}(t) + b/2$.

Next, we turn to the PDEs for $h^{(i)}$. From (\ref{eq:hi-1order}), we choose $h^{(i)}_0 = h^{(i)}_0(t,\bfmu)$ independent of $y$. Therefore, the PDE (\ref{eq:hi-12order}) reduces to $\cL_0 h^{(i)}_1 = 0$. Consequently,  for $i=0,1$, we choose $h^{(i)}_1 = h^{(i)}_1(t,\bfmu)$ also independent of $y$, thus $\cL_1 h^{(i)}_1 = 0$. Equations (\ref{eq:h10order}) and (\ref{eq:h00order}) are Poisson equations for $h^{(1)}_2$ and  $h^{(0)}_2$, respectively, and their centering conditions are
\small
\begin{align*}
\left\langle (\partial_t + \cL_{\bfmu}) h^{(1)}_0 + \bfgamma \cdot \bfmu + \frac{\cX_0(t) h^{(1)}_0}{k(t,\cdot)} \right\rangle = 0
\quad
&\Leftrightarrow
\quad
(\partial_t + \cL_{\bfmu}) h^{(1)}_0 + \bfgamma \cdot \bfmu + \frac{\cX_0(t)}{\kappa(t)} h^{(1)}_0 = 0,\\
\left\langle (\partial_t + \cL_{\bfmu}) h^{(0)}_0 + \frac{(h^{(1)}_0)^2}{4k(t,\cdot)} \right\rangle = 0
\quad
&\Leftrightarrow
\quad
(\partial_t + \cL_{\bfmu}) h^{(0)}_0 + \frac{1}{4\kappa(t)}(h^{(1)}_0)^2 = 0.
\end{align*}
\normalsize
These PDEs for $h^{(0)}_0$ and $h^{(1)}_0$ result from the analysis in Section \ref{sec:deterministic_k} with $\mathsf{k} = \kappa$ (specifically corresponding to \eqref{eqn:h0-deterministic-PDE}
and \eqref{eqn:h1-deterministic-PDE}), and, hence, $h^{(0)}_0 = h^{(0)}$ and $h^{(1)}_0 = h^{(1)}$ as in equations \eqref{eq:h0cj} and \eqref{eq:h1cj}, respectively.
\end{proof}

\begin{proposition}[First-order Terms]
The Poisson equations (\ref{eq:h112order}) and (\ref{eq:h012order}) are centered when $h^{(1)}_1$ and $h^{(0)}_1$ solve the PDEs
\begin{align}
(\partial_t + \cL_{\bfmu}) h^{(1)}_1 + \frac{\cX_0(t)}{\kappa(t)} h^{(1)}_1 + \bfV \cdot F^{(1)}(t,\bfmu) = 0, \label{eq:pde_h11}\\
(\partial_t + \cL_{\bfmu}) h^{(0)}_1  + \frac{h^{(1)}_0}{2\,\kappa(t)}\left(h^{(1)}_1 + \bfV \cdot F^{(0)}(t,\bfmu)\right) = 0, \label{eq:pde_h01}
\end{align}
where
\begin{align}
&\bfV = -\bfrho\, \langle \beta \psi'\rangle, \quad F^{(1)}(t,\bfmu) = \frac{\cX_0(t)}{\kappa(t)} \sum_{i=1}^k g_i(\bfmu) \partial_{\mu_i} h^{(1)}_0(t,\bfmu),
\quad \text{and}
\label{eqn:F1}
\\
&F^{(0)}(t,\bfmu) = \sum_{i=1}^k g_i(\bfmu)\partial_{\mu_i} h^{(1)}_0(t,\bfmu),
\label{eqn:F0}
\end{align}
and $\psi$ is the solution of the centered Poisson equation
\begin{align}\label{eq:psi}
\cL_0 \psi(y) = \eta(y).
\end{align}
\end{proposition}

\begin{proof}

From the Poisson equations (\ref{eq:h10order}) and (\ref{eq:h00order}), and accounting for their centering constants, which may depend on time and trading signal, we find
\begin{align}
h^{(1)}_2(t,\bfmu,y) &= - (\cL_0^{-1} \eta(y)) \left(\frac{\cX_0(t)}{\kappa(t)}h^{(1)}_0(t,\bfmu) \right) + c^{(1)}(t,\bfmu),
\label{eqn:h21-inv-L0}
\\
h^{(0)}_2(t,\bfmu,y) &= - (\cL_0^{-1} \eta(y)) \left(\frac{1}{4\kappa(t)}(h^{(1)}_0(t,\bfmu))^2 \right) + c^{(0)}(t,\bfmu),
\label{eqn:h20-inv-L0}
\end{align}
for some functions $c^{(1)}$ and $c^{(0)}$ independent of $y$. Due to our standing assumption $\langle \eta \rangle = 0$, the Poisson equation (\ref{eq:psi}) is centered and we can use $\psi$ to simplify \eqref{eqn:h21-inv-L0} \eqref{eqn:h20-inv-L0}. Hence, we may write
\begin{align}
h^{(1)}_2(t,\bfmu,y) &= - \psi(y)\, \frac{\cX_0(t)}{\kappa(t)}h^{(1)}_0(t,\bfmu) + c^{(1)}(t,\bfmu),
\\
h^{(0)}_2(t,\bfmu,y) &= - \psi(y)\, \frac{1}{4\kappa(t)}(h^{(1)}_0(t,\bfmu))^2 + c^{(0)}(t,\bfmu).
\end{align}
Therefore,
\begin{align}
\cL_1 h^{(1)}_2(t,\bfmu,y) &= - \psi'(y)  \beta(y) \frac{\cX_0(t)}{\kappa(t)} \sum_{i=1}^k (g_i(\bfmu) \cdot \bfrho )\partial_{\mu_i} h^{(1)}_0(t,\bfmu),\\
\cL_1  h^{(0)}_2(t,\bfmu,y) &= -\psi'(y)  \beta(y) \frac{h^{(1)}_0(t,\bfmu)}{2\kappa(t)}  \sum_{i=1}^k (g_i(\bfmu) \cdot \bfrho )\partial_{\mu_i} h^{(1)}_0(t,\bfmu),
\end{align}
and the proposition follows.
\end{proof}

The solution to these PDEs (\ref{eq:pde_h11}) and (\ref{eq:pde_h01}) are provided in the following proposition. Its proof is a straightforward application of the Feynman-Kac representation.
\begin{proposition}\label{prop:G1}
The unique solutions to (\ref{eq:pde_h01}) and (\ref{eq:pde_h11}) are
\small
\begin{align}
\label{eqn:h1-phi1}
h^{(1)}_1(t,\bfmu) &:= \bfV \cdot \phi_1(t,\bfmu), \quad \text{and}
\\
h^{(0)}_1(t,\bfmu) &:= \bfV \cdot\int_t^T   \bE\left[\ \left. \tfrac{h^{(1)}_0(s,\bfmu_s)}{2\kappa(s)}\left(\phi_1(s,\bfmu_s) + F^{(0)}(t,\bfmu_s)\right)   \ \right| \ \bfmu_t = \bfmu \ \right]ds,
\end{align}
\normalsize
where
\begin{align}\label{eq:phi1}
\phi_1(t,\bfmu) = \int_t^T e^{\int_t^s \frac{\cX_0(u)}{\kappa(u)} du}\; \bE[\ F^{(1)}(s,\bfmu_s) \  | \ \bfmu_t = \bfmu\ ]\ ds,
\end{align}
\seb{
where $F^{(0)}$ and $F^{(1)}$ are defined in \eqref{eqn:F0} and \eqref{eqn:F1}, respectively.
}
\end{proposition}


\subsection*{First-order approximation}

Putting all these results together, we can find the first-order approximation for $H^\eps$ is given by
\begin{align} \label{eq:first_order_approx}
H^\eps_1(t,x,S,\bfmu,q) = x + q\,S + h(t,\bfmu,q) + \bfV^\eps \cdot h_1(t,\bfmu,q),
\end{align}
where $\bfV^\eps = \sqrt{\eps}\;\bfV$, $h_1(t,\bfmu,q) = h^{(1)}_1(t,\bfmu)\,q + h^{(0)}_1(t,\bfmu)$ and $h(t,\bfmu,q)$ is given by Equation \eqref{eq:h_deterministic} with $\mathsf{k} = \kappa$.

\subsection{Approximation for the optimal feedback control}

\yuri{Under the deterministic temporary price impact, there are two classical trading strategies. Firstly, the approach outlined by Almgren-Chriss (AC), where no trading signal is considered, in our setting becomes
\begin{align}
\nu^{AC}(t,q) = \frac{\cX_0(t)}{\kappa(t)}q.
\end{align}
Including trading signals (TS), we are in the setting of Section \ref{sec:deterministic_k} and we define
\begin{align}
\nu^{TS}(t,\bfmu,q) = \frac{1}{\kappa(t)}\left(\cX_0(t)\, q + \frac{1}{2} h^{(1)}_0(t,\bfmu)\right) = \nu^{AC}(t,q) + \frac{1}{2\kappa(t)} h^{(1)}_0(t,\bfmu) . \label{eq:v_0}
\end{align}
We refer to $\nu^{AC}$ and $\nu^{TS}$ by AC and TS optimal controls, respectively.

Let us now consider the connection of these formulas with the one found in our full model. Recall that the optimal feedback control is
\begin{align}
\nu_*^\eps(t,\bfmu,q,y) = \tfrac{1}{2\,k(t,y)} (b\,q + \partial_q h^\eps) = \tfrac{1}{2\kappa(t)} (1 + \eta(y)) \left(bq + \partial_q h^\eps \right)\,.
\end{align}
Using the expansion derived for $h^\eps$, we find, formally,
\begin{subequations}
\begin{align}
\nu_*^\eps(t,\bfmu,q,y) &= \frac{1}{2\,k(t,y)} \left(b\,q + h^{(1),\eps} + 2\,q \,h^{(2),\eps} \right)
\\
&= \frac{1}{2\,k(t,y)} \left(2 \,\cX^\eps\, q + h^{(1),\eps} \right) = \frac{\cX^\eps}{k(t,y)}\,q  + \frac{1}{2\,k(t,y)}\, h^{(1),\eps}
\\
&= \frac{\cX_0(t)}{k(t,y)}\,q + \frac{1}{2\,k(t,y)} \left(h^{(1)}_0(t,\bfmu) + \sqrt{\eps} \,h^{(1)}_1(t,\bfmu)\right) + \cdots
\\
&= \frac{\cX_0(t)}{k(t,y)}\,q  +\frac{1}{2k(t,y)} \,h^{(1)}_0(t,\bfmu) + \frac{1}{2\,k(t,y)}\, \sqrt{\eps}\, h^{(1)}_1(t,\bfmu) + \cdots.
\end{align}%
\end{subequations}%
The first-order approximation of the optimal control is then given by the expression
\begin{subequations}
\begin{align}
\nu_1^\eps(t,\bfmu,q,y) &= (1 + \eta(y))\; \left(\nu^{TS}(t,\bfmu,q) + \sqrt{\eps}\; \frac{1}{2\kappa(t)}\,h^{(1)}_1(t,\bfmu) \right)\\
&= (1 + \eta(y))\; \left(\nu^{TS}(t,\bfmu, q) + \bfV^\eps \cdot C_1(t,\bfmu) \right)\label{eqn: v1}
\end{align}
\end{subequations}
where
\begin{align}\label{eqn: C1}
C_1(t,\bfmu) = \frac{1}{2\kappa(t)} \phi_1(t,\bfmu),
\end{align}
with $\phi_1$ given in \eqref{eq:phi1}. Thus, the correction term is proportional to the function $\phi_1$, which in turn is the expected weighted average of the future sensitivity of the zero-order value to the trading signal, see Equations \eqref{eq:phi1}, \eqref{eqn:F1} and \eqref{eq:h1cj}.

Notice that, since the process $Y$ is observable in our setting, we are able to consider the first-order control $\nu_1^\eps$ in practice. In the numerical examples below, we will compare $\nu^{AC}$, $\nu^{TS}$ and $\nu_1^\eps$. Note that $(1 + \eta(y))\nu^{TS}$ becomes the zeroth-order approximation. 
\begin{remark}
If one computes the first-order approximation of the solution to \textit{linear} PDE that arises from PDE \eqref{eq:pdf_h_eps} but taking $v = \nu_0(t,\boldsymbol{\mu},q,y) = (1 + \eta(y))\nu^{TS}(t,\boldsymbol{\mu},q)$ instead of the supremum, then one could show that this is also a first-order approximation of $h^\varepsilon$.
\end{remark}
}

\section{A Mean-Reverting Trading Signal Example}
\label{sec:MR-trading-signal}

In this section, we  use a specific model to show how the results of the previous section may be applied in practice. To completely specify the model, we must determine the dynamics of $\bfmu$ and $Y^\eps$. For the latter, we choose $\alpha(y) = \theta - y$ and $\beta(y) := \sqrt{2}\,\beta$ (for a constant $\beta$). Thus, $Y$'s stationary  distribution is Gaussian with mean $\theta$ and variance $\beta^2$. Setting  $\eta(y) = y$, we must have $\theta = 0$ to guarantee $\langle \eta \rangle = 0$. In this case, we find $\psi(y) = -y$ (which solves \eqref{eq:psi}), and hence
$$
\bfV^\eps = \sqrt{2\,\eps}\,\beta\,\bfrho.
$$
We model the trading signal $\bfmu$ as a multidimensional Ornstein-Uhlenbeck (OU) process
$$d\bfmu_t = (A\,\bfmu_t + \bar{\bfmu}) \,dt + B\, d\bfW_t,$$
where $A, B \in \bR^{k \times k}$ and $\bar{\bfmu} \in \bR^k$.
Straightforward computations  show that, for $s \geq t$,
$$\bfmu_s = e^{A(s-t)}\, \bfmu_t + \int_t^s e^{A(s-u)}\, du \ \bar{\bfmu}  + \int_t^s e^{A(s-u)}\, B\, d\bfW_u.$$
Then
$$\bE[\ \bfmu_s \ | \ \bfmu_t = \bfmu  \ ]
= e^{A(s-t)} \ \bfmu + \int_t^s e^{A(s-u)} \ du \ \bar{\bfmu},$$
and from \eqref{eq:h1cj}, we may write (for the case of deterministic impact)
\begin{align}
\label{eqn: hcj-for-mr}
h^{(1)}(t,\bfmu) &= \bfgamma \cdot \left( \int_t^T e^{\int_t^s \frac{\cX(u)}{k(u)} du} \left( e^{A(s-t)} \bfmu + \int_t^s e^{A(s-u)} du \ \bar{\bfmu} \right) ds\right)\\
&= \bfgamma \cdot\left(\Phi_1(t)\bfmu + \Phi_0(t) \bar{\bfmu}\right),
\end{align}
where
\begin{equation}
\Phi_1(t) = \int_t^T e^{\int_t^s \frac{\cX(u)}{k(u)} du} e^{A(s-t)}\, ds \quad \mbox{ and } \quad  \Phi_0(t) = \int_t^T e^{\int_t^s \frac{\cX(u)}{k(u)} du} \int_t^s e^{A(s-u)} \,du \,ds.
\label{eqn:Phi0-Phi1-defn}
\end{equation}
If $A$ is invertible, then we may write $\Phi_0$ in terms of $\Phi_1$ using
$$
\int_t^s e^{A(s-u)} du = (e^{A(s-t)} - I)\, A^{-1}
$$
which implies
\begin{align*}
\Phi_0(t) &= \int_t^T e^{\int_t^s \frac{\cX(u)}{k(u)} du} (e^{A(s-t)} - I) A^{-1} ds
= \Phi_1(t) - \int_t^T e^{\int_t^s \frac{\cX(u)}{k(u)} du} ds A^{-1}.
\end{align*}

Using the representation in \eqref{eqn: hcj-for-mr}, the explicit formula for the \yuri{TS} optimal trading speed is
\begin{align}
\label{eqn:v0-explicit-mr}
\nu^{TS}(t,\bfmu,q) &= -\frac{\cX(t)}{\kappa(t)}  q - \frac{1}{2\kappa(t)} \bfgamma \cdot (\Phi_1(t)\, \bfmu + \Phi_0(t) \,\bar{\bfmu}).
\end{align}

Next,  we analyze $\phi_1$ (see \eqref{eq:phi1}) to obtain the first-order correction. First note that
$$
F^{(1)}(t,\bfmu) = \frac{\cX(t)}{\kappa(t)} \sum_{i=1}^k \bfb_i  \partial_{\mu_i} h^{(1)}(t,\bfmu),
$$
where $\bfb_i$ is the $i$-th row of $B$. Moreover,
$$\partial_{\mu_i}h^{(1)}(t,\bfmu) = \sum_{j=1}^k \gamma_j \Phi_1(t)_{ji} = \bfgamma \cdot \Phi_{1,i}(t),$$
where $\Phi_{1,i}(t)$ is the $i$-th column of $\Phi_1(t)$. Thus
$$F_1(t,\bfmu) = \frac{\cX(t)}{\kappa(t)} \sum_{i=1}^k (\bfgamma \cdot \Phi_{1,i}(t)) \bfb_i,$$
which is independent of $\bfmu$. Therefore, $\phi_1$ is independent of $\bfmu$ and
\begin{align*}
\phi_1(t) &= \int_t^T e^{\int_t^s \frac{\cX(u)}{\kappa(u)} du} F_1(s) \,ds
\\
&= \int_t^T e^{\int_t^s \frac{\cX(u)}{\kappa(u)} du} \,\frac{\cX(s)}{\kappa(s)}\,
\sum_{i=1}^k (\bfgamma \cdot \Phi_{1,i}(s)) \,\bfb_i \,ds
= \sum_{i=1}^k (\bfgamma \cdot \Phi_{2,i}(t)) b_i,
\end{align*}
where
$$
\Phi_{2,ji}(t) = \int_t^T e^{\int_t^s \frac{\cX(u)}{\kappa(u)} du}\, \frac{\cX(s)}{\kappa(s)} \,\Phi_{1,ji}(s)\, ds.$$
Hence, $C_1$ is independent of $\bfmu$ and
$$
C_1(t) = \frac{1}{2 \kappa(t)}\, \phi_1(t) = \frac{1}{2 \kappa(t)} \sum_{i=1}^k \bfb_i \,(\bfgamma \cdot \Phi_{2,i}(t)).
$$

Inserting these results into \eqref{eqn: v1} and \eqref{eqn:v0-explicit-mr}, provides the first-order approximation for the optimal strategy as
\begin{align*}
\nu_1^\eps(t,\bfmu,q) &= \frac{\cX(t)}{\kappa(t)}  q + \frac{1}{2\kappa(t)}\, \bfgamma \cdot (\Phi_1(t)\,\bfmu + \Phi_0(t)\, \bar{\bfmu})
+ \frac{1}{2 \kappa(t)} \sum_{i=1}^k (\bfV^\eps \cdot \bfb_i)\, (\bfgamma \cdot \Phi_{2,i}(t))\,.
\end{align*}

\section{Real Data Example}
\label{sec:real-data}

In this section, we  present an application of our model  with real data from the Microsoft stock \yuri{on May 24th, 2014.}

The impact parameter is estimated using the cost of walking the limit order book by `fictitious orders' of various volumes, and regressing this cost/volume curve to obtain a linear approximation, \yuri{see \cite{CarteaJaimungal2015} for a more detailed description of the procedure.} The slope is our estimate of $\kappa_t = k(t,Y_t^\eps)$, and we obtain this estimate every second of the day, resulting in a sample path as shown in \cref{fig:kappa_fit}. To obtain estimates for $\eta_t = \eta(Y_t^\eps)$, we perform the following steps:
\begin{enumerate}

    \item Project the sample path of $\kappa_t$  onto a polynomial basis of order $J=8$, i.e.  find
    \[
    \boldsymbol{\alpha}^\star:=\arg\min_{\boldsymbol{\alpha}} \textstyle\sum_{i=1}^{N} \left(\kappa_{t_i}-\sum_{j=1}^J \alpha_j \,{t_i}^{j-1} \right)^2,
    \]
    where $N=23400$ (number of seconds in a trading day), $t_i=i/N$.

    \item Adjust the coefficients to ensure that  $\langle\eta_t\rangle=0$. We do this by making the empirical mean of the implied $\eta_t$ process over the data as close to zero as possible
    \[
    \boldsymbol{\alpha}^*:=\arg\min_{\boldsymbol{\alpha}} \tfrac{1}{N}\textstyle\sum_{i=1}^{N} \left(\frac{\sum_{j=1}^J \alpha_j {t_i}^{j-1}}{\kappa_{t_i}} - 1\right)
    \]
    where the initial starting point for the minimizer is $\boldsymbol{\alpha}^\star$ from step 1.

    \item Next, we assume $\eta(y) = y$ and set $ Y_t^\varepsilon = \eta_t =\frac{\sum_{j=1}^J \alpha_j {t}^{j-1}}{\kappa_{t}} - 1 $, and assume $Y_t^\varepsilon$ is an OU process and hence satisfies
    \[
    dY_t^\varepsilon = -\tfrac{1}{\varepsilon}Y_t^\varepsilon\,dt + \tfrac{1}{\sqrt{\varepsilon}}\,\beta\,dW_t^*\,.
    \]

    \item Finally, we estimate the parameters $\varepsilon$ and $\beta$ by regressing $(\eta_{t_{i+1}}-\eta_{t_i})$ onto $\eta_{t_i}$.

\end{enumerate}

\seb{The results of the estimation are provided in Tables \ref{tab:estimated_param} and \ref{tab:exo_param}. The trading signal parameters in the left portion of Table \ref{tab:exo_param} are similar to those used in \cite{lehalle2019incorporating}.} The $\eta_t$ sample path, together with scatter plot of $\eta_t$ and $\eta_{t-1}$ and histogram of residuals are shown in Figure \ref{fig:eta}. \seb{The estimate of $\eps$ we perform is equivalent to the variogram approach espoused in, e.g., \cite{fouque2000derivatives}. Moreover, as the estimated $\eps\ll1$, we conclude there is a fast mean-reverting factor at play here.}

\begin{table}[h!]
\color{black}
\footnotesize
    \centering
    \begin{tabular}[t]{lc}
    Param. & Value \\
    \hline
    $\eps$  & 0.0035 \\
            &  (0.0030, 0.0043) \\
    $\beta$ & 0.26984 \\
            & (0.26981, 0.26986) \\
    \hline
    \end{tabular}
    \caption{\seb{Estimated parameters using the procedure described above. The confidence intervals are reported at the 95\% level, and for $\eps$ it is computed using the confidence interval of the slope of the regression in step 4 and the one for $\beta$ is computed using the Delta method.}}
    \label{tab:estimated_param}
\end{table}

\begin{table}[h!]
\color{black}
\footnotesize
    \centering
    \begin{tabular}[t]{lc}
    Param. & Value \\
    \hline
    $A$ & -10 \\
    $\bar{\mu}$ & 0 \\
    $B$ & 1 \\
    $\mu_0$ & 1 \\
    $S_0$ & 100\\
    \hline
    \end{tabular}
    \quad
        \begin{tabular}[t]{lc}
    Param. & Value \\
    \hline
    $\sigma$ & 0.1\\
    $\rho$ & -0.5 \\
    $\gamma$ & 0.1 \\
    $V_1^\eps$ & -0.008 \\
    $b$ & 1.4275e-06 \\
    \hline
    \end{tabular}
    \caption{\seb{Typical values for the parameters that conclude the dynamics of $(S,\mu,Y)$. The parameter $\rho$ is the correlation between the trading signal and the fast mean-reverting process and $V_1^\eps = \sqrt{2\eps}\beta \rho$. Moreover, the value of $b$ is estimated using the procedure in \cite{CarteaJaimungal2015} on MSFT data for the calendar year 2014.}}
    \label{tab:exo_param}
\end{table}

\begin{table}[h!]
\color{black}
\footnotesize
    \centering
    \begin{tabular}[t]{lc}
    Param. & Value \\
    \hline
    $X_0$ & 0 \\
    $Q_0$ & $10^4$ \\
    $\phi$ & $b$ \\
    $\varphi$ & $10^3 \cdot b$ \\
    \hline
    \end{tabular}
    \caption{\seb{Typical values for the parameters for the agent's criteria.}}
    \label{tab:agent_param}
\end{table}

\begin{figure}[t]
    \centering
    \hfill
    \includegraphics[width=.32\textwidth]{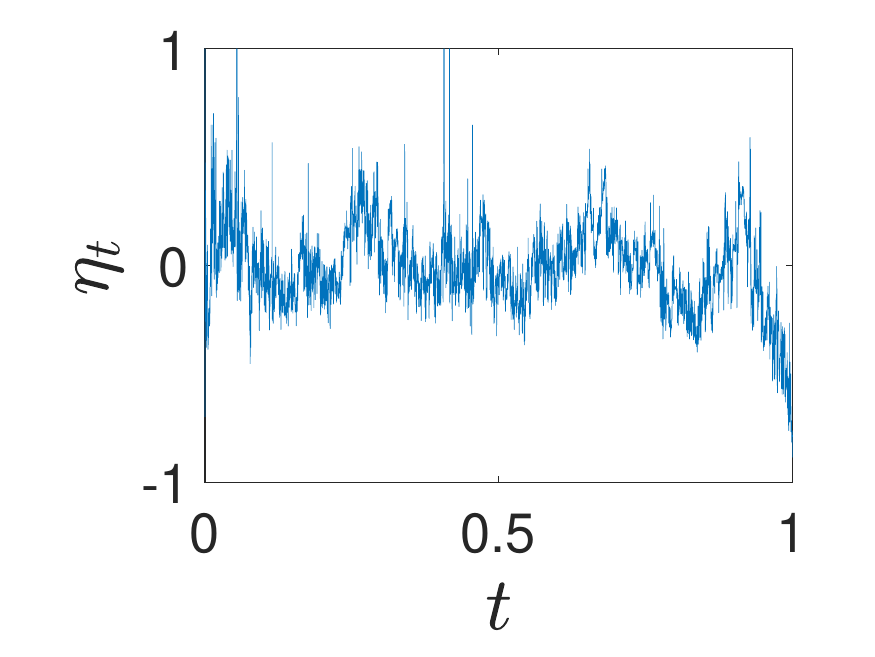}
    \hfill
    \includegraphics[width=.32\textwidth]{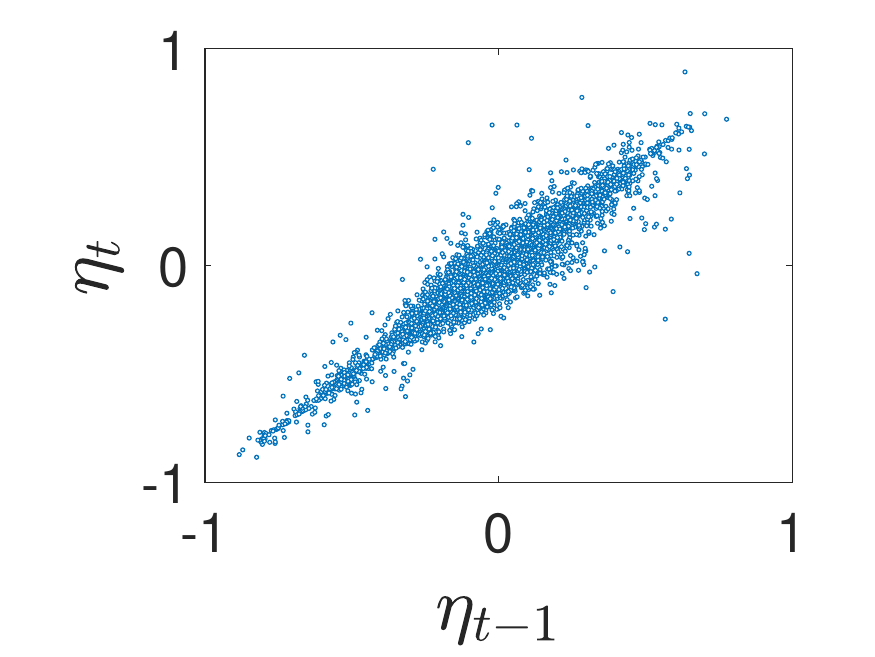}
    \hfill
    \includegraphics[width=.32\textwidth]{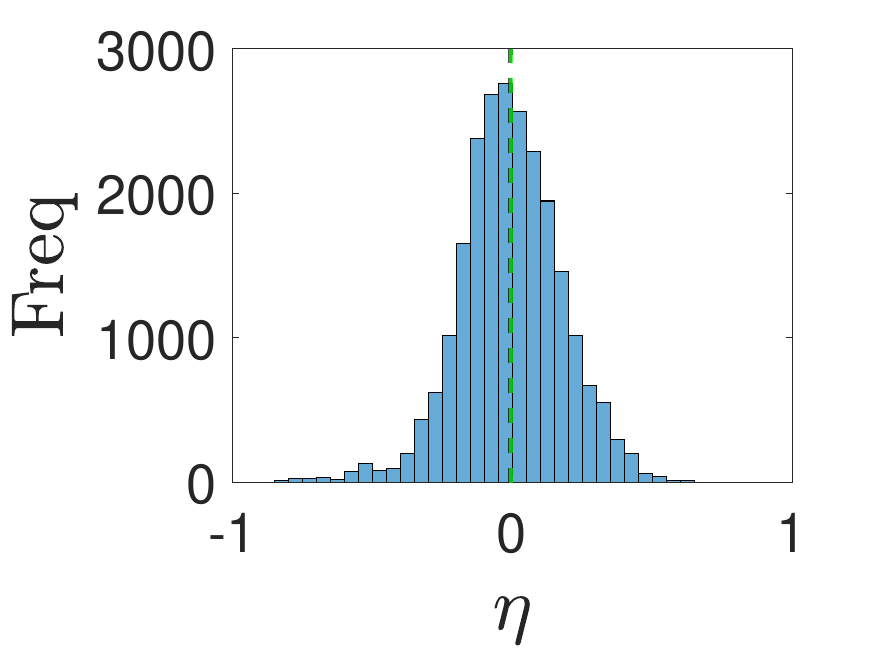}
    \hfill
    \caption{Estimation of the process $Y^\eps$ and of the stationary distribution.}
    \label{fig:eta}
\end{figure}

Figure \ref{fig:controls} illustrates the key functions ($\cX(t)$ which solves the ODE \eqref{eqn:Ricatti-deterministic}, $\Phi_0$ and $\Phi_1$ given in \eqref{eqn:Phi0-Phi1-defn}, and $C_1$ defined in \eqref{eqn: C1}) that feed into \yuri{TS (trading signals)} and first-order control strategies, $\nu^{TS}$ and $(1 + \eta(y))\nu_1^\eps$, using the estimated parameters. The figure shows the results for three different running penalty parameters $\phi=b$, $5\,b$, and $10\,b$. \seb{The reason we choose the urgency parameter to be of the same order of magnitude as $b$ is because (i) the solution to \eqref{eqn:Ricatti-deterministic} when $\mathsf{k}(t)$ is constant, depends on the ratio of the urgency parameter to the temporary impact $\frac{\phi}{\mathsf{k}}$; and (ii) as shown in \cite{cartea2017algorithmic}, the temporary and permanent impact are of similar order of magnitude. Thus, it is this ratio that modulates the agent's urgency, rather than the absolute value of $\phi$ alone.} As $\phi$ increases, we see that $\cX(t)$ generally increases in magnitude (more negative), $\Phi_0$ and $\Phi_1$ decreases, while $C_1$ behaves non-monotonically.
\begin{figure}[h!]
    \centering
    \includegraphics[width=0.7\textwidth]{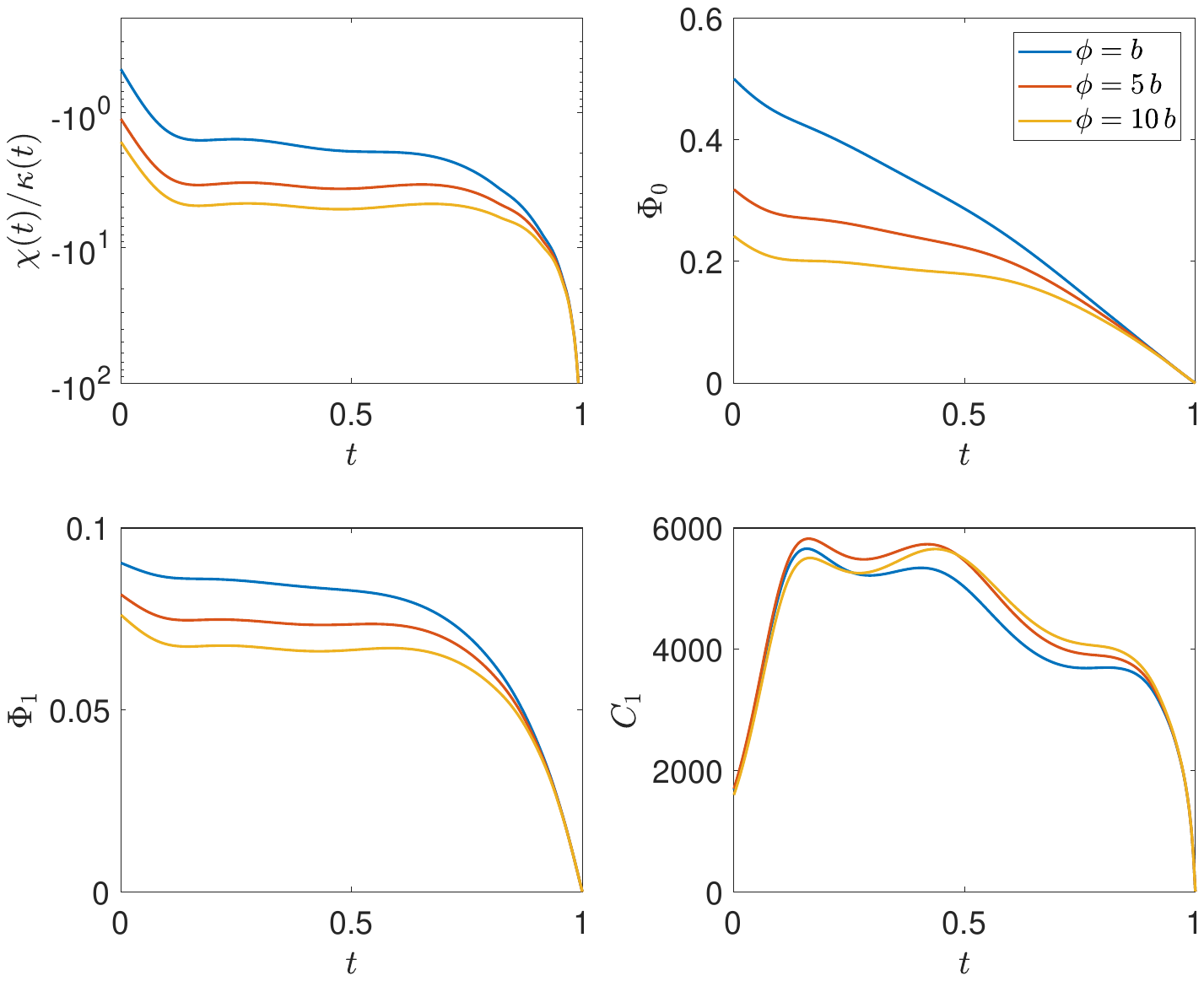}
    \caption{Plot of $\cX/\kappa$, $\Phi_0$, $\Phi_1$ and $C_1$ that define the control strategies $\nu^{TS}$ and $\nu_1^\eps$ for various levels of $\phi$.}
    \label{fig:controls}
\end{figure}

Figure \ref{fig:paths} shows a sample path of the zero-order and first-order optimal controls for various running penalties. The figure shows that the starting level of trading speed increases with $\phi$, as the trader is more urgent to rid themselves of shares, but larger $\phi$ levels induce the trader to slow down sooner in the day, and ends trading more slowly.  As well, and as expected, the figure shows that responding to the stochastic impact on average traces out the \yuri{TS} path.
\begin{figure}[h!]
    \centering
    \begin{minipage}{0.32\textwidth}
    \centering
    \tiny \textbf{$\phi=b$}\par
    \includegraphics[width=\textwidth]{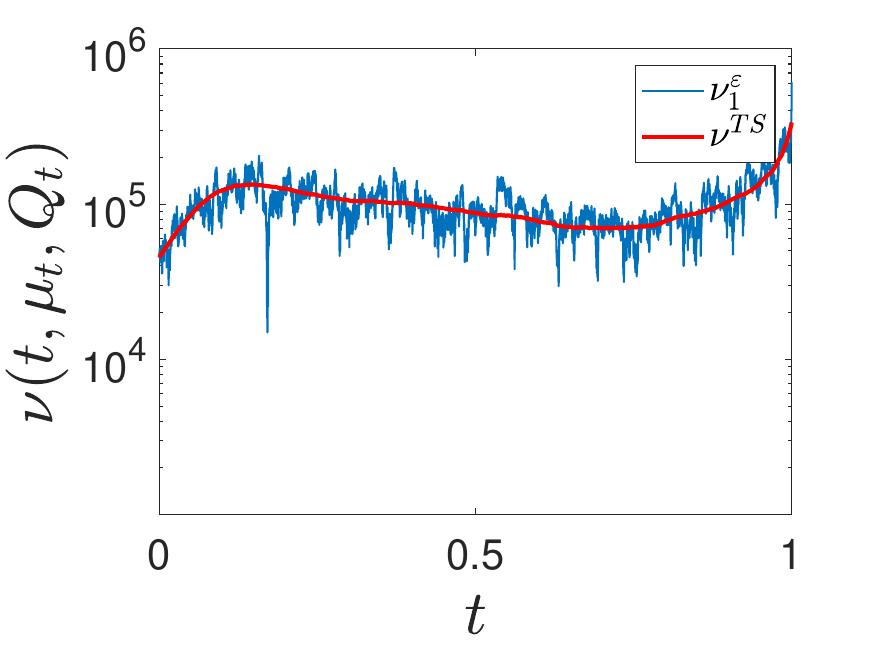}
    \end{minipage}
    \begin{minipage}{0.32\textwidth}
    \centering
    \tiny
    \textbf{$\phi=5\,b$}\par
    \includegraphics[width=\textwidth]{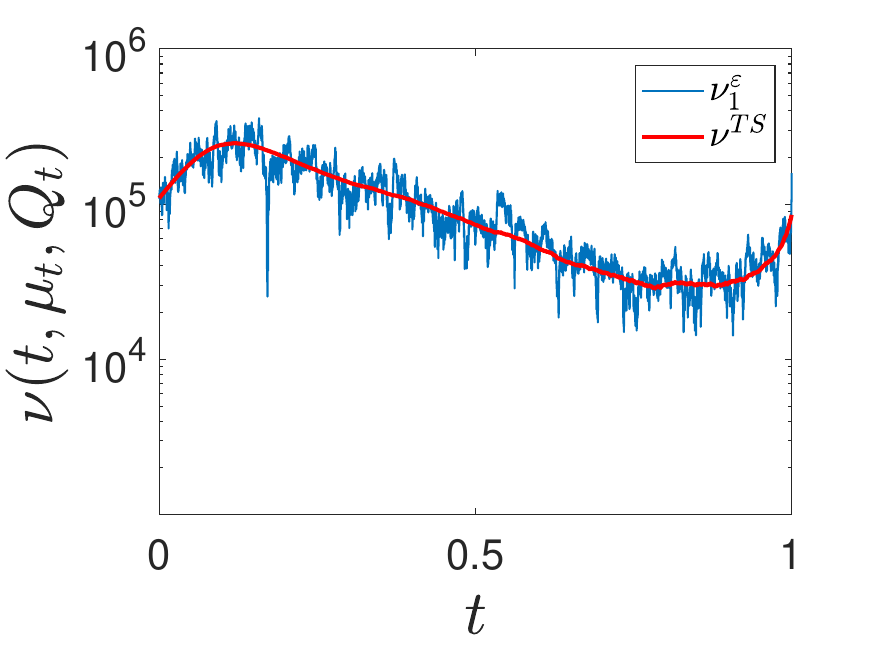}
    \end{minipage}
    \begin{minipage}{0.32\textwidth}
    \centering
    \tiny
    \textbf{$\phi=10\,b$}\par
    \includegraphics[width=\textwidth]{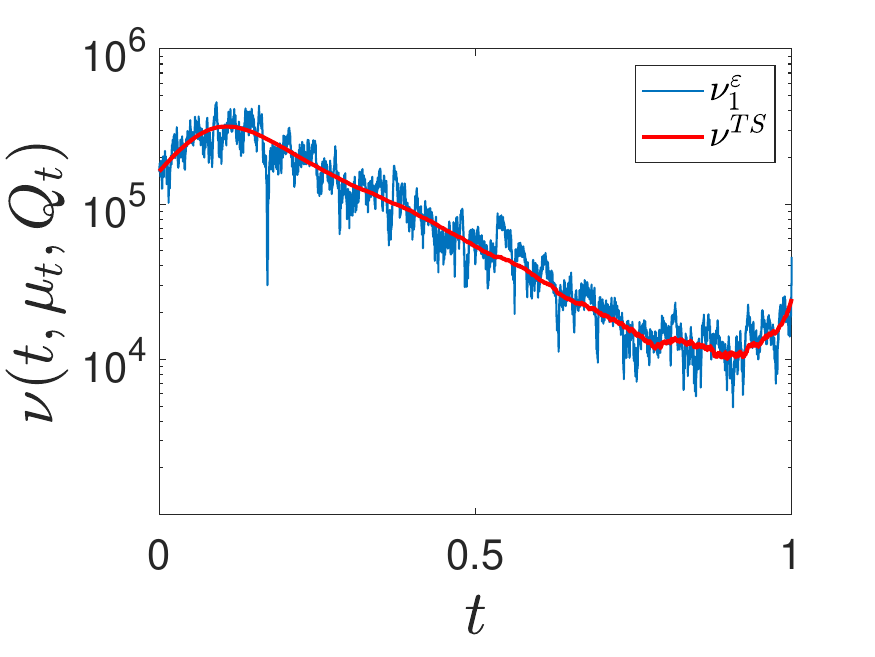}
    \end{minipage}
    \caption{A sample path of the \yuri{TS} (red) and first-order (blue) optimal controls, with $\phi=b,\;5\,b\;,10\,b$.}
    \label{fig:paths}
\end{figure}

Next, we simulate 10,000 sample paths of the \yuri{TS} and first-order optimal controls, and the corresponding inventory paths. Figure \ref{fig:Q-compare-paths} shows how the first-order optimal control's inventory paths differs from (left) the AC paths -- defined as the path with trading speed $\nu^{AC}$ -- and (right) the zero-order paths. The figure shows a single sample path (generated with the same set of random numbers, although they do produce different impacts due to differing trading speeds), together with the 10\%, 50\%, and 90\% quantiles across all 10,000 sample paths. As the figure shows, for the median path, there is a slight slow down relative to the AC path (since the median has an upward bump). As urgency ($\phi$) increases, the deviations develop a hump just after the start of trading and towards the end of trading the deviations are pinned more closely to the AC solution. The deviation from the \yuri{TS} strategy have a median path of essentially zero, and, as expected, the variation around the zero-order path is smaller than around the AC path. There is, however, still a significant amount of deviation due to the first-order strategy correctly adjusting to the stochastic impact.
\begin{figure}
\begin{minipage}[t]{0.6\textwidth}
    \centering
    {\small $\phi=b$}
    \par
    \includegraphics[width=\textwidth]{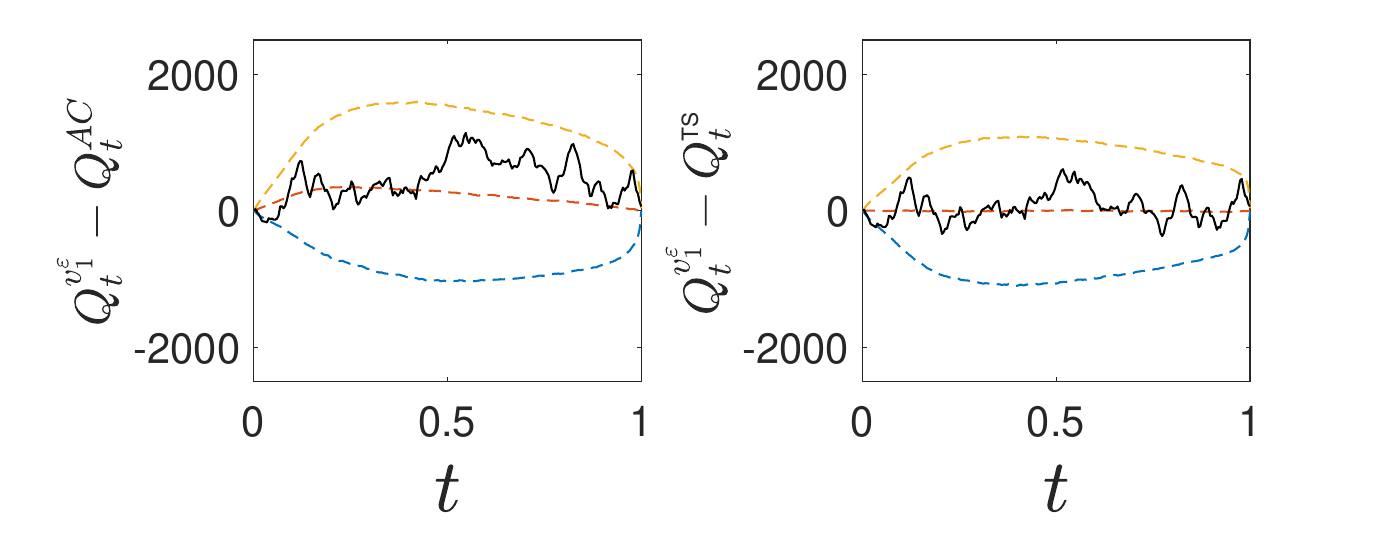}
    \\
    {\small $\phi=5\,b$}
    \par
    \includegraphics[width=\textwidth]{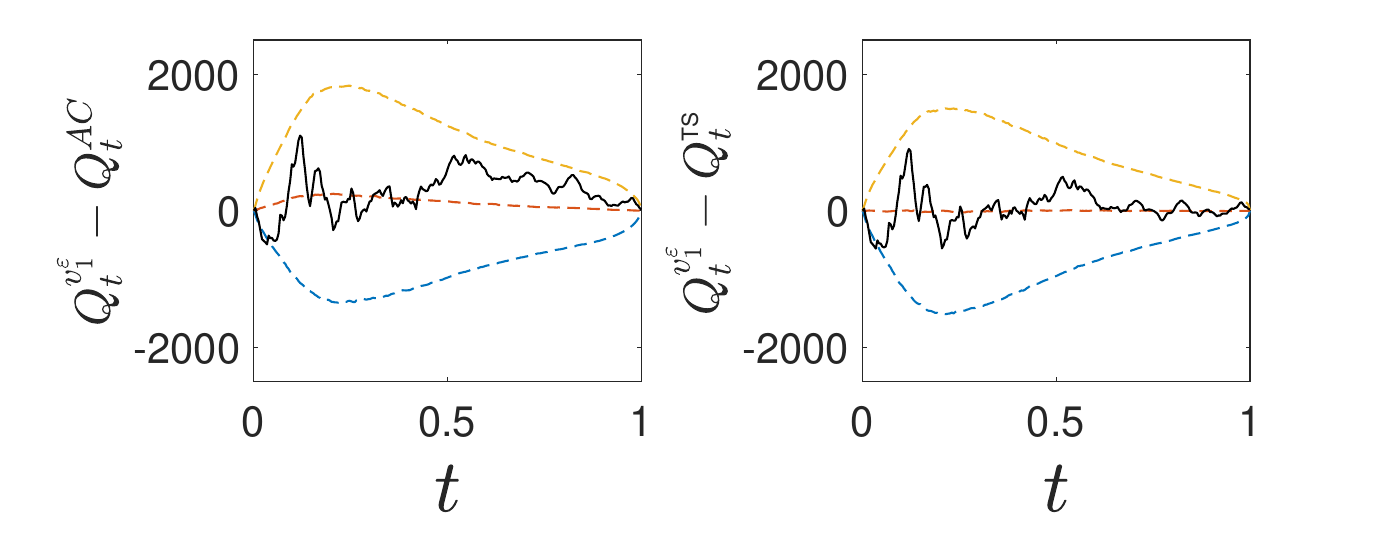}
    \\
    {\small $\phi=10\,b$}
    \par
    \includegraphics[width=\textwidth]{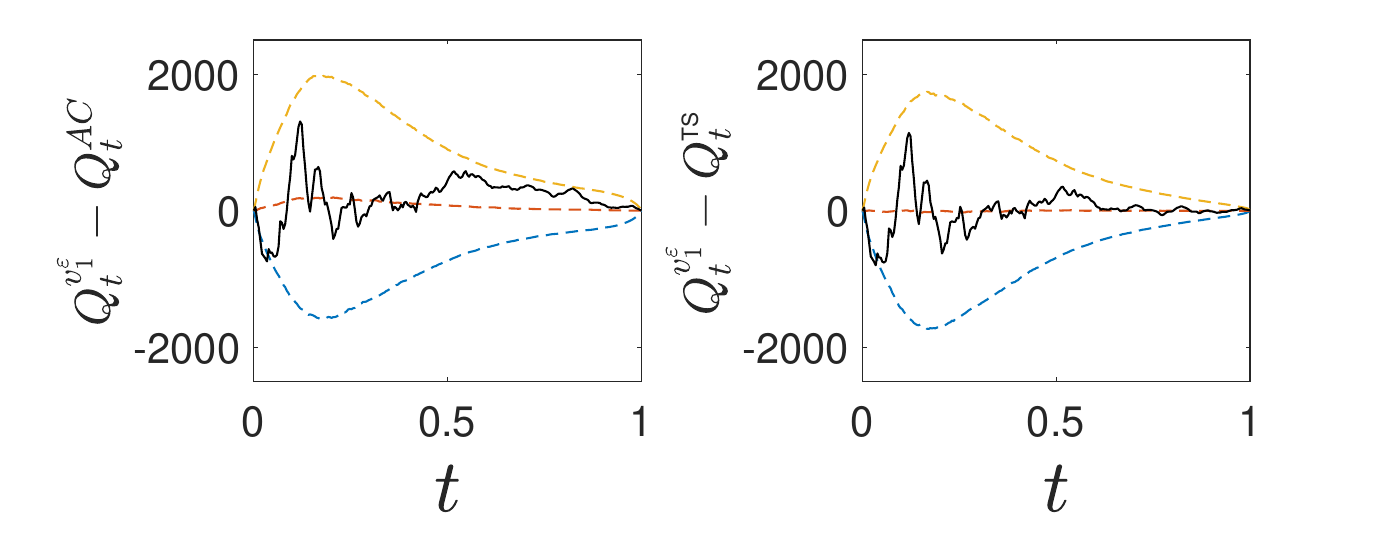}
\end{minipage}
\begin{minipage}[t]{0.35\textwidth}
    \caption{A sample path and the 10\%, 50\%, and 90\% quantile of (left) the deviation of the inventory following the first-order optimal and the Almgren-Chriss controls, and (right) the deviation of the inventory following the first-order optimal and the \yuri{TS} optimal controls.}
    \label{fig:Q-compare-paths}
\end{minipage}
\end{figure}

As a final numerical comparison, we investigate the cost savings that the first-order strategy provides over the AC and the zero-order strategy. For this purpose we compute the cost $C_T^\nu = X_T^\nu + Q_T^\nu S_T^\nu$ associated with a strategy $\nu$ and compute the savings in basis points relative to a benchmark $\bar{\nu}$ (with $\bar{\nu}$ either given by $\nu^{AC}$ or $\nu^{TS}$) as
\begin{equation}
    \text{basis points savings} = \frac{C_T^\nu - C_T^{\bar{\nu}}}{C_T^{\bar{\nu}}}\times 10^{4}\,.
\end{equation}
The results for various levels of urgency are shown in Figure \ref{fig:basis_pt}. The results show the  as we increase the urgency we on average perform better than both the AC and zero-order strategy. \seb{ It is worth mentioning that, while our optimisation problem includes the urgency penalty, the urgency penalty is effectively a form of risk control for the agent, and hence when comparing the results of a particular optimal strategy, it makes economic sense to compare the `real' costs associated with the strategy -- which excludes the urgency penalty -- and that is what Figure \ref{fig:basis_pt} compares.
}

\begin{figure}[!h]
    \centering
    \includegraphics[width=.45\textwidth]{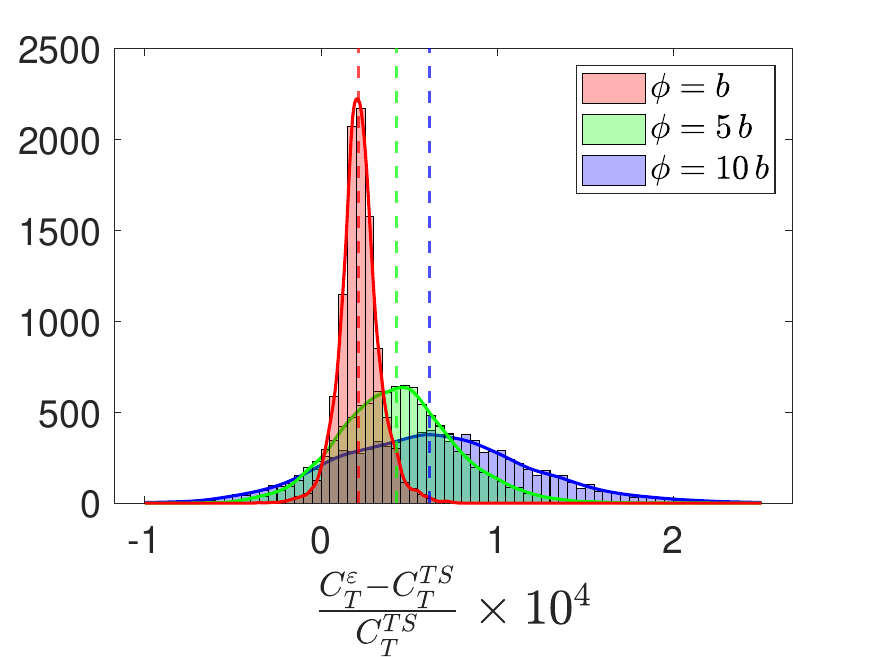}
    \quad
    \includegraphics[width=.45\textwidth]{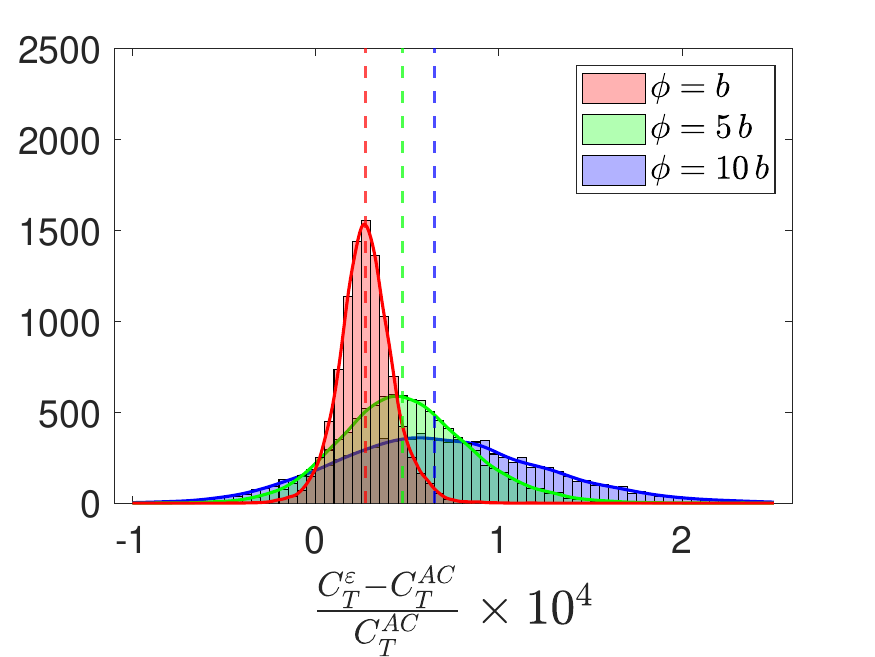}
    \caption{Histogram of the savings in basis points for various $\phi$ between first-order strategy and \yuri{TC and AC strategies. Dashed lines show the location of the corresponding median. To simplify notation, we use $C^\eps$ instead of $C^{\nu_1^\eps}$.}}
\label{fig:basis_pt}
\end{figure}

\section{Accuracy of the Approximation}\label{sec:approx_proof}

In this section we  prove the accuracy of the approximation of $\cX^\eps$ by $\cX_0$. As noted earlier, once the accuracy for $h^{(2),\eps}$ is proved, the accuracy for the approximations for $h^{(1),\eps}$ and $h^{(0),\eps}$, which are the solutions of the \textit{linear} PDEs (\ref{eq:pde_h1eps}) and (\ref{eq:pde_h0eps}), follow from the usual arguments from \cite{multiscale_fouque_new_book}. The proof of accuracy for non-linear PDEs was first developed in \cite{fouque2020nonlinearaccuracy} under the different context of portfolio optimization.

Recall that $\cX^\eps$ satisfies the non-linear PDE (\ref{eq:pde_chieps}):
\begin{align}
\left\{
\begin{array}{rl}
\partial_t \cX^\eps - \phi + \frac{(\cX^\eps)^2}{k(t,y)} + \frac{1}{\eps} \cL_0  \cX^\eps &= 0,
\\
\cX^\eps(T,y) &= - \varphi + b/2.
\end{array}
\right. \label{eq:pde_chieps_proof}
\end{align}
The zero-order approximation $\cX_0$ is the solution of Ricatti ODE
\begin{align}
\left\{
\begin{array}{rl}
\displaystyle \cX'_0(t) - \phi + \frac{1}{\kappa(t)}\cX^2_0(t) &= 0,\\
\cX_0(T) &= -\varphi + b/2.
\end{array}
\right.
\end{align}
To provide precise bounds, we make the following standing assumption for the remainder of this section.\\

\begin{minipage}[c]{0.9\textwidth}
\begin{assumptions}\label{assump:proof}
We assume that
\begin{enumerate}
    \item $\kappa$ is bounded and bounded away from zero with bounded derivative; and
    \item $\eta$ is bounded, implying that the solutions to all Poisson equations considered below and their derivatives are bounded.\footnotemark
\end{enumerate}
\end{assumptions}
\end{minipage}
\footnotetext{For the proof of Theorem \ref{thm:accuracy}, it is only required that the solutions to all Poisson equations considered below and their derivatives are bounded. To this effect, we would additionally need the source of these Poisson equations to decay at the boundary of the domain so we can ensure this behavior for its solution. For instance, this would be ensured if the source has compact support.}

\begin{theorem}\label{thm:accuracy}
With Assumptions \ref{assump:proof} enforced,  there exists $C > 0$ such that, for any $\eps < 1$,
$$|\cX^\eps(t,y) - \cX_0(t)| \leq C\, \eps,$$
for any $t \in [0,T]$ and $y \in \bR$.
\end{theorem}
\begin{proof}
The proof involves three main steps: (i) recast the PDE \eqref{eq:pde_chieps} for $\cX^\eps$ in terms of an auxiliary control problem, (ii) construct a sub-solution and a super-solution of \eqref{eq:pde_chieps}, and (iii) use the sub/super-solutions to bound the error.

First, the PDE for $\cX^\eps$ may be recast as the following HJB equation
\begin{align}\label{eq:pde_chi_eps}
\left\{
\begin{array}{rl}
\partial_t \cX^\eps - \phi + \sup_{\xi \in \cB} \left\{-2\cX^\eps \xi - k(t,y)\xi^2 \right\} + \frac{1}{\eps} \cL_0 \cX^\eps &= 0,\\ \\
\cX^\eps(T,y) &= - \varphi + \tfrac{b}{2},
\end{array}
\right.
\end{align}
where $\cB = \bR$. Therefore, non-linear Feynman-Kac (see \cite{pham_control}) implies that $\cX^\eps$ admits the representation in terms of the auxiliary control problem
\begin{align}
\cX^\eps(t,y) &= \sup_{\xi \in \bB} Z(t,y,\xi),\\
Z(t,y,\xi) &= \bE_{t,y}\left[\;(- \varphi + \tfrac{b}{2})\, e^{-2\int_t^T \xi_udu} + \int_t^T e^{-2\int_t^s \xi_udu} \,(-\phi - k(s,Y_s^\eps)\, \xi_s^2)\,ds\;\right],
\end{align}
where $\bB$ is a set of admissible controls taking values in $\cB$ that consists of $\mathcal{F}$-predictable processes  in $\mathbb{L}^2([0,T],\Omega)$.
We next proceed along the ideas developed in \cite{fouque2020nonlinearaccuracy} for the proof of accuracy of non-linear PDEs stemming from  portfolio optimisation problems.

To this end, define $\cX^\pm$ as follows
\begin{align}
\cX^\pm(t,y) &= \cX_0(t) + \eps \,\tilde{\cX}_1(t,y) \pm \eps \,(2T - t) \,C_1 \,N(t) \pm \eps^2\, M(t,y),
\label{eqn:def-xi-pm}
\end{align}
for constant $C_1 > 0$ and functions $\tilde{\cX}_1$, $N$ and $M$ to be defined shortly.

By the Poisson equation \eqref{eq:chi0order}, we find
\begin{align}
\cX_1(t,y) = - \psi(y) \frac{\cX_0^2(t)}{\kappa(t)} + c(t),
\end{align}
for some function $c$ independent of $y$ and $\psi$ is the solution of the Poisson equation $\cL_0 \psi = \eta$ which is defined only up to a deterministic function of time (constant in $y$). By choosing this function, we may therefore assume $\langle \psi \rangle = 0$ without loss of generality. 
We define the $y$-dependent part of $\cX_1$ as
\begin{align}
\tilde{\cX}_1(t,y) := - \psi(y)\, \frac{\cX_0^2(t)}{\kappa(t)},
\end{align}
which sets one of the three free functions in the definition of $\cX^\pm$.

Next, define the differential operator
\begin{align}
R^\xi[\cX] = \partial_t \cX - \phi - 2\xi\cX - k(t,y)\xi^2 + \tfrac{1}{\eps}\cL_0\cX,
\end{align}
for any $\xi \in \cB$. The HJB equation \eqref{eq:pde_chi_eps} may be written as $\sup_{\xi \in \cB}R^\xi[\cX] = 0$.
For any control $\xi \in \bB$, It\^o's lemma implies
\begin{align}
&\hspace*{-2em}\cX^\pm(T,Y_T^\eps)\, e^{-2\int_t^T \xi_u du}  - \int_t^T e^{-2\int_t^s \xi_u du}\, (\phi + k(s,Y_s^\eps)\, \xi_s^2)\,ds
\label{eq:ito}
\\
=&\, \cX^\pm(t,y) + \int_t^T e^{-2\int_t^s \xi_u du} \,R^{\xi_s}[\cX^\pm](s,Y_s^\eps) \,ds
\\
&+ \frac{1}{\sqrt{\eps}} \int_t^T
e^{-2\int_t^s \xi_u du} \,\beta(Y_s^\eps) \,\partial_y \cX^\pm(s,Y_s^\eps) \, dW_s^*.
\end{align}
Taking expectation, since the It\^o integral above is a true martingale by Assumptions \ref{assump:proof}, we find
\begin{align}\label{eq:Z_proof}
Z(t,y,\xi) + \eps\, \cH^\pm(t,y,\xi) = \cX^\pm(t,y) + \int_t^T \bE_t\left[\, e^{-2\int_t^s \xi_udu}\, R^{\xi_s}[\,\cX^\pm](s,Y_s^\eps) \,\right] \,ds,
\end{align}
\begin{align}
\cH^\pm(t,y,\xi) := \bE_{t,y}\left[\,e^{-2\int_t^T \xi_udu} \,\left(\tilde{\cX}_1(T,Y_T^\eps) \pm T \,C_1 N(T) \pm \eps\, M(T,Y_T^\eps)\right)\,\right]\,.
\end{align}
Next, we construct particularly choices for the two remaining free functions  of $\cX^\pm$, defined in \eqref{eqn:def-xi-pm}, that provide sub and super solutions.

\subsection*{Sub-solution}

Let us first analyze $\cX^-$. Acting on it with the $R^\xi$ operator, we have
\begin{equation}
\begin{split}
R^\xi[\cX^-] =& \cX_0' - \phi -2 \xi \cX_0 - k(t,y) \xi^2  + \tfrac{1}{\eps}\cancelto{0}{\cL_0\cX_0}
\\
&-  \eps \partial_t \tilde{\cX}_1 -2 \eps\xi \tilde{\cX}_1 - \cL_0 \tilde{\cX}_1
\\
&+ \eps C_1 N(t) - \eps C_1 (2T - t) \left(\partial_t N - 2\xi N \right) \\
&- \eps^2(\partial_t M - 2\xi M) - \eps \cL_0 M.
\end{split}
\end{equation}
Define the particular auxiliary control
\begin{equation}
\xi_0(t,y) := - \frac{\cX_0(t)}{k(t,y)} = -(1 + \eta(y))\frac{\cX_0(t)}{\kappa(t)}.
\end{equation}
With this control, we have that
\begin{align}
\cX_0' - \phi -2 \xi_0 \cX_0 - k(t,y) \xi_0^2 = \cX_0' - \phi + \tfrac{\cX_0^2}{k(t,y)}\,,
\end{align}
hence, by \eqref{eq:chi0order}
\begin{align}
\cX_0' - \phi + \tfrac{\cX_0^2}{k(t,y)} + \cL_0 \tilde{\cX}_1 = 0.
\end{align}
Therefore,
\begin{align}
\begin{split}
R^{\xi_0}[\cX^-] =&\,\eps\, \big\{\partial_t \tilde{\cX}_1 -2  \xi_0 \tilde{\cX}_1 + C_1 N(t)
- C_1 (2T - t) \left(N'(t) - 2\xi_0 N \right) - \cL_0 M\big\}
\\
&- \eps^2\,\left\{\partial_t M - 2\xi_0 M\right\}
\end{split}
\\
\begin{split}
=& \, \eps\,
\Big\{-\psi(y) \left(\partial_t (\cX^2/\kappa) + 2 \cX^3/\kappa^2 \right) -2 \eta(y)\psi(y) \cX^3/\kappa^2 + C_1 N(t)
\\
&\quad - C_1 (2T - t)  \left(N'(t) + 2\cX_0 N/\kappa \right) - C_1 (2T - t) \eta(y) \cX_0 N/\kappa - \cL_0 M\Big\}
\\
&- \eps^2\,\left\{\partial_t M - 2\xi_0 M\right\}.
\end{split}
\end{align}
We next choose $M$ such that the $\eps$ term is centered with respect to the invariant measure, that is
\begin{align}
M(t,y) &= -\Psi_1(y) \left(\partial_t (\cX_0^2/\kappa) + 2\cX_0^3/\kappa^2\right) - 2\Psi_2(y) \cX^3/\kappa^2 
- \psi(y) C_1 (2T - t) \cX_0 N /\kappa,
\end{align}
where we use that $\langle \eta \rangle = \langle \psi \rangle = 0$, and where
\begin{align}
\cL_0 \Psi_1 = \psi \quad \mbox{and}\quad  \cL_0 \Psi_2 = \eta \psi - \langle \eta \psi \rangle.
\end{align}
Therefore, we find
\begin{equation}
\begin{split}
R^{\xi_0}[\cX^-] =&\, \eps \left\{-2 \langle \psi \eta \rangle \frac{\cX^3}{\kappa^2} + C_1 N(t)
- C_1 (2T - t) \left(N'(t) + 2\frac{\cX_0\,N}{\kappa} \right)\right\}
\\
&- \eps^2\left\{\partial_t M - 2\xi_0 M\right\}.
\end{split}
\end{equation}
Next, we choose $N$ such that the term proportional to $(2T-t)$ vanishes, which is given by
\begin{align}
N(t) = \exp\left\{-2\int_0^t \frac{\cX_0(s)}{\kappa(s)}ds\right\}.
\end{align}
For the above choices of $N$ and $M$, we therefore obtain
\begin{equation}
R^{\xi_0}[\cX^-] =\eps \left(-2 \langle \psi \eta \rangle \cX^3/\kappa^2 + C_1 N(t) \right) - \eps^2(\partial_t M - 2\xi_0 M).
\end{equation}
By Assumptions \ref{assump:proof}, the first and last terms of the expression above are bounded. Hence, there exists $C_1$ large enough such that
\begin{align}
\int_t^T \bE_{t,y} \left[ e^{-2\int_t^s \xi_0(u,Y_u^\eps) du} \;R^{\xi_0}[\cX^-](s,Y_s^\eps) \right] ds \geq 0,
\end{align}
for any $t$ and $y$. 
Moreover, notice that
\begin{align*}
\cH^-(t,y,\xi_0)   = \bE_{t,y}\left[e^{-2\int_t^T \xi_0(u,Y_u^\eps) du} \left(\tilde{\cX}_1(T,Y_T^\eps) - T\,C_1\, N(T) - \eps\, M(T,Y_T^\eps)\right)\right].
\end{align*}
Again by Assumptions \ref{assump:proof}, if necessary, we increase $C_1$ such that the term inside the expectation that defines $\cH^-(t,y,\xi_0)$ is negative everywhere. Hence, by the definition of $\cX^\eps$ we find that, for any $(t,y) \in [0,T] \times \bR$,
\begin{align}
\cX^\eps(t,y) &\geq Z(t,y,\xi_0)
\\
&= \cX^-(t,y) - \eps \cH^-(t,y,\xi_0)
+ \int_t^T \bE_{t,y} \left[ e^{-2\int_t^s \xi_0(u,Y_u^\eps) du} R^{\xi_0}[\cX^+](s,Y_s^\eps) \right] ds \\
&\geq \cX^-(t,y).
\end{align}

\subsection*{Super-solution}

Let us now analyze $\cX^+$. We start by computing
\begin{align}
\sup_{\xi \in \cB} R^\xi[\cX^+] &= \partial_t \cX^+ - \phi + \frac{(\cX^+)^2}{k(t,y)} + \frac{1}{\eps} \cL_0 \cX^+
\\
\begin{split}
=& \eps\,\Big\{\partial_t \tilde{\cX}_1 + \frac{2}{k(t,y)} \cX_0 (\tilde{\cX}_1 + (2T-t) C_1 N + \eps M) \\
&  -C_1N + (2T-t) C_1\partial_t N  + \cL_0 M\Big\}
\\
& + \eps^2 \left\{\frac{1}{k(t,y)}(\tilde{\cX}_1 + (2T-t)C_1N + \eps M)^2 + \partial_t M\right\}.
\end{split}
\end{align}
The term of order $\eps$ is then
\begin{align}\label{eq:super_center}
\partial_t \tilde{\cX}_1 &+ \frac{2}{k(t,y)} \cX_0 (\tilde{\cX}_1 + (2T - t)C_1N)  -C_1N \\
&+ (2T - t)C_1\partial_t N + \cL_0M. \nonumber
\end{align}
Similarly to what was done before, we choose $M$ to center (\ref{eq:super_center}) with respect to the invariant distribution. Since $\langle \psi \rangle = \langle \eta \rangle = 0$, the term of order $\eps$ becomes
\begin{align}
-2 \langle \eta \psi \rangle \cX_0^3/\kappa^2  -C_1N + (2T - t)C_1(\partial_t N + 2\cX_0N/\kappa).
\end{align}
Choosing the same function $N$ as for the sub-solution, the term of order $\eps$ is given by
\begin{align}
-2 \langle \eta \psi \rangle \cX_0^3/\kappa^2  -C_1N.
\end{align}
By Assumptions \ref{assump:proof}, the term of order $\eps^2$ is bounded. We can then choose $C_1$ large enough so that $\sup_{\xi \in \cB} R^\xi[\cX^+] \leq 0$.
Furthermore,
\begin{align}
\cX^+(T,y) = (-\varphi + b/2) + \eps \tilde{\cX}_1(T) + \eps T C_1 N(T) + \eps^2 M(T,y).
\end{align}
Possibly increasing $C_1$ even more, we conclude that, say for $\eps \leq 1$, $\cX^+(T,y) \geq -A + b/2$ uniformly in $y$. Then, by It\^o's formula (\ref{eq:ito}),
\small
\begin{align}
&Z(t,y,\xi) \leq \bE_{t,y}\left[\cX^+(T,Y_T^\eps) e^{-2\int_t^T \xi_u du}  + \int_t^T e^{-2\int_t^s \xi_u du} (-\phi - k(s,Y_s^\eps) \xi_s^2)d \right]\\
&= \cX^+(t,y) +  \int_t^T \bE_{t,y} \left[ e^{-2\int_t^s \xi_udu} R^{\xi_s}[\cX^+](s,Y_s^\eps) \right] ds\\
&\leq \cX^+(t,y) +  \int_t^T \bE_{t,y} \left[ e^{-2\int_t^s \xi_udu} \sup_{\xi \in \cB} R^\xi[\cX^+](s,Y_s^\eps) \right] ds \leq \cX^+(t,y).
\end{align}
\normalsize
Therefore,
\begin{align}
\cX^\eps(t,y) \leq \cX^+(t,y).
\end{align}

We have finally concluded that there exists a constant $C > 0$ such that
\begin{align}
|\cX^\eps(t,y) - \cX_0(t)| \leq C \eps.
\end{align}
\end{proof}

\appendix

\section{Auxiliary Results for Deterministic Impact}
%
\label{app:proofs_time_dep}

We present here the proof of the formulas presented in Section \ref{sec:deterministic_k}. We follow the steps of \cite{CarteaJaimungal2015}. Write $h(t,\bfmu,q) = h^{(0)}(t,\bfmu) +  h^{(1)}(t,\bfmu) q +  h^{(2)}(t,\bfmu)q^2$, with $h^{(0)}(T,\bfmu) = h^{(1)}(T,\bfmu) = 0 $ and $h^{(2)}(T,\bfmu) = -\varphi$. Therefore, we find the following PDEs for $h^{(i)}$:
\begin{align}
&\left(\partial_t + \cL_{\bfmu} \right)h^{(2)}(t,\bfmu) - \phi  + \frac{1}{4\mathsf{k}(t)}(b + 2h^{(2)}(t,\bfmu))^2 = 0,
\label{eqn:h2-deterministic-PDE}
\\[1em]
&\left(\partial_t + \cL_{\bfmu} \right)h^{(1)}(t,\bfmu) + \bfgamma \cdot \bfmu  + \frac{1}{2\mathsf{k}(t)}(b + 2h^{(2)}(t,\bfmu))h^{(1)}(t,\bfmu) = 0,
\label{eqn:h1-deterministic-PDE}
\\[1em]
&\left(\partial_t + \cL_{\bfmu} \right)h^{(0)}(t,\bfmu) + \frac{1}{4\mathsf{k}(t)}(h^{(1)})^2(t,\bfmu) = 0.
\label{eqn:h0-deterministic-PDE}
\end{align}

To solve the first equation above, notice that we may assume $h^{(2)} = h^{(2)}(t)$ independent of $\bfmu$. Define
\begin{align}\label{eq:chi}
\cX(t) = \frac{b}{2} + h^{(2)}(t).
\end{align}
and notice that
$$\cX'(t) - \phi + \frac{1}{\mathsf{k}(t)} \cX^2(t) = 0.$$
Now, $h_1$ satisfies
\begin{align}
\left(\partial_t + \cL_{\bfmu} \right)h^{(1)}(t,\bfmu) + \bfgamma \cdot \bfmu  + \frac{\cX(t)}{\mathsf{k}(t)}h^{(1)}(t,\bfmu) = 0.
\end{align}
By Feynman-Kac's representation,
\begin{align}
h^{(1)}(t,\bfmu) &= \bE\left[\int_t^T e^{\int_t^s \frac{\cX(u)}{\mathsf{k}(u)} du} (\bfgamma \cdot \bfmu_s) ds \ \Big| \ \bfmu_t = \bfmu \right] \\
&=  \int_t^T e^{\int_t^s \frac{\cX(u)}{\mathsf{k}(u)} du} \bE[\bfgamma \cdot \bfmu_s \  | \ \bfmu_t = \bfmu] ds, \\
h^{(0)}(t,\bfmu) &= \bE\left[\int_t^T  \frac{1}{4\mathsf{k}(s)}(h^{(1)})(s,\bfmu_s) ds \ \Big| \ \bfmu_t = \bfmu \right]\\
&=\int_t^T \frac{1}{4\mathsf{k}(s)} \bE\left[  (h^{(1)})^2(s,\bfmu_s) \ | \ \bfmu_t = \bfmu \right]ds.
\end{align}

\vspace{10pt}




\bibliographystyle{siamplain}
\bibliography{main}

\end{document}